\documentclass[twocolumn,amsmath,amssymb,aps,superscriptaddress,nofootinbib,letter]{revtex4}

\usepackage{graphicx}
\usepackage{dcolumn}
\usepackage{bm}
\usepackage{color}
\usepackage{amsthm}
\usepackage{amsfonts}
\usepackage{bbm}
\usepackage{epsfig}
\usepackage{xcolor}
\usepackage[citebordercolor=red,urlbordercolor=red,linkbordercolor=blue]{hyperref}

\newcommand{\R}{\mathbb{R}}
\newcommand{\C}{\mathbb{C}}
\newcommand{\G}{\mathcal{G}}
\newcommand{\N}{\mathbb{N}}

\newtheorem{theorem}{Theorem}
\newtheorem{definition}[theorem]{Definition}
\newtheorem{lemma}[theorem]{Lemma}
\newtheorem{proposition}[theorem]{Proposition}
\newtheorem{postulate}{Postulate}

\DeclareMathOperator{\lin}{\mathrm{lin}}
\DeclareMathOperator{\im}{\mathrm{im}}
\newcommand{\ray}{\mathrm{Ray}}
\def\intersect{\cap}
\def\meet{\wedge}
\def\join{\vee}
\newcommand{\beq}{\begin{equation}}
\newcommand{\eeq}{\end{equation}}
\newcommand{\beqa}{\begin{eqnarray}}
\newcommand{\eeqa}{\end{eqnarray}}
\newcommand{\g}{{\mathfrak g}}
\newcommand{\ce}{{\cal E}}
\newcommand{\tr}{\mbox{tr}}

\begin{document}

\title{Higher-order interference and single-system postulates characterizing quantum theory}

\author{Howard Barnum} \email{hnbarnum@aol.com}
\affiliation{Department of Physics and Astronomy, University of
  New Mexico, 1919 Lomas Blvd. NE, Albuquerque, NM 87131}
\affiliation{Stellenbosch Institute for Advanced Studies (STIAS),
  Wallenberg Research Center at Stellenbosch University, Marais
  Street, Stellenbosch 7600, South Africa}
\author{Markus P.\ M\"uller} \email{markus@mpmueller.net}
\affiliation{Institut f\"ur Theoretische Physik, Universit\"at Heidelberg, Philosophenweg 19, D-69120 Heidelberg, Germany}
\author{Cozmin Ududec} \email{cozster@gmail.com}
\affiliation{Invenia Technical Computing, 135
    Innovation Dr., Winnipeg, MB R3T 6A8, Canada}

\date{April 8, 2014; revised November 12, 2014}

\begin{abstract}
We present a new characterization of quantum theory in terms of simple
physical principles that is different from previous ones in two
important respects: first, it only refers to properties of single
systems without any assumptions on the composition of many systems; and second, it
is closer to experiment by having absence of higher-order interference
as a postulate, which is currently the subject of experimental
investigation. We give three postulates -- no higher-order
interference, classical decomposability of states, and strong symmetry
-- and prove that the only non-classical operational probabilistic theories
satisfying them are real, complex, and quaternionic quantum theory,
together with $3$-level octonionic quantum theory and ball state
spaces of arbitrary dimension.  Then we show that adding observability
of energy as a fourth postulate yields complex quantum theory as the
unique solution, relating the emergence of the complex numbers to the
possibility of Hamiltonian dynamics.  We also show that there may be
interesting non-quantum theories satisfying only the first two of our
postulates, which would allow for higher-order interference in
experiments while still respecting the contextuality analogue of the
local orthogonality principle.
\end{abstract}

\pacs{Valid PACS appear here}

\maketitle

\section{Introduction} \label{SecIntroduction}

Quantum theory currently underpins much of modern physics and is
essential in many other scientific fields and countless technological
applications.  However, by most accounts quantum phenomena remain
rather mysterious: there is no generally accepted intuitive picture of
the underlying reality, and the standard textbook introductions of the
mathematical formalism lack a simple conceptual motivation.

With the rise of quantum information processing and the ever more
refined control of quantum phenomena, there has recently been a surge of diverse attempts to tackle such foundational questions.  These range from studies
of the information processing capabilities of theories similar to
quantum theory~\cite{BBLW08clone, BBLW07broadcast, BBLW08tele, 
GMCDrevdyn, AaronsonIsland}, to reconstructions of the formalism
from information-theoretic principles~\cite{Hardy01, DakicBrukner, MMAxioms, GiulioDer11, Fivel2010, InformationUnit}, to no-go theorems regarding
interpretations and generalizations of the formalism~\cite{PBR, CabelloSeveriniWinter, ColbeckRenner}, to novel experiments testing
various predictions of the theory~\cite{Sinha2010, SollnerWeihs12, GeorgeEtAl13}.

In this paper we give several closely related reconstructions of the
mathematical structure---Hilbert space, Hermitian observables, positive
operator-valued measures---of finite-dimensional quantum theory from
simple postulates with clear physical significance and generality.

Providing such an explanation for the Hilbert space
  structure of quantum theory in terms of physically (not just
mathematically) natural postulates is important for several reasons.
First, deeper and more reasonable principles can help to dissolve the
mysteries of quantum phenomena and make them more intelligible and
easier to teach.  Two well-known examples of this approach are Kepler's laws of planetary motion and their explanation through Newton's laws of motion and gravitation, and the Lorentz transformations and their explanation in Einstein's two relativity postulates. Second, it can be argued that this approach will be essential
in making progress on problems such as formulating a theory
unifying quantum and gravitational physics, as well as
for developing potentially more accurate and more fundamental
theories.  In the absence of a picture of the underlying reality, we
can use first principles to proceed toward the next physical theory in
a careful, conceptual fashion.  More practically, this approach can
shed light on what is responsible for the power of quantum information
processing and cryptography.

Because quantum theory applies to an extremely broad
range of physical systems and phenomena, and its probabilistic
structure seems essential, we work within a
broad framework for studying probabilistic physical theories (usually
called \emph{operational probabilistic theories}).  These are theories
that succinctly describe sets of
experiments and assign probabilities to measurement outcomes.  More
precisely, we imagine that physicists, or nature, prepare physical
systems in various states, and then observe these systems in various
ways.  The outcomes of these observations occur with certain
probabilities, which are predicted by the theory.  It is important to
emphasize that we do not assume that these probabilities are described
by quantum theory; instead our postulates will allow us to derive
their structure as represented by quantum theory.

Our postulates are as follows:
\begin{enumerate}
\item \textbf{Classical Decomposability}: Every state of a physical system
can be represented as a probabilistic mixture of perfectly distinguishable states of maximal knowledge (``pure states'').
\item \textbf{Strong Symmetry}: Every set of perfectly distinguishable
  pure states of a given size can be reversibly transformed to any
  other such set of the same size.
\item \textbf{No Higher-Order Interference}: The interference pattern
  between mutually exclusive ``paths'' in an experiment is
  exactly the sum of the patterns which would be observed in all
  two-path sub-experiments, corrected for overlaps.
\item \textbf{Observability of Energy}: There is non-trivial
continuous reversible time evolution, and the generator of every such evolution
can be associated to an observable (``energy'') which is a conserved quantity.
\end{enumerate}

\begin{figure}[!hbt]
\begin{center}
\includegraphics[angle=0, width=6cm]{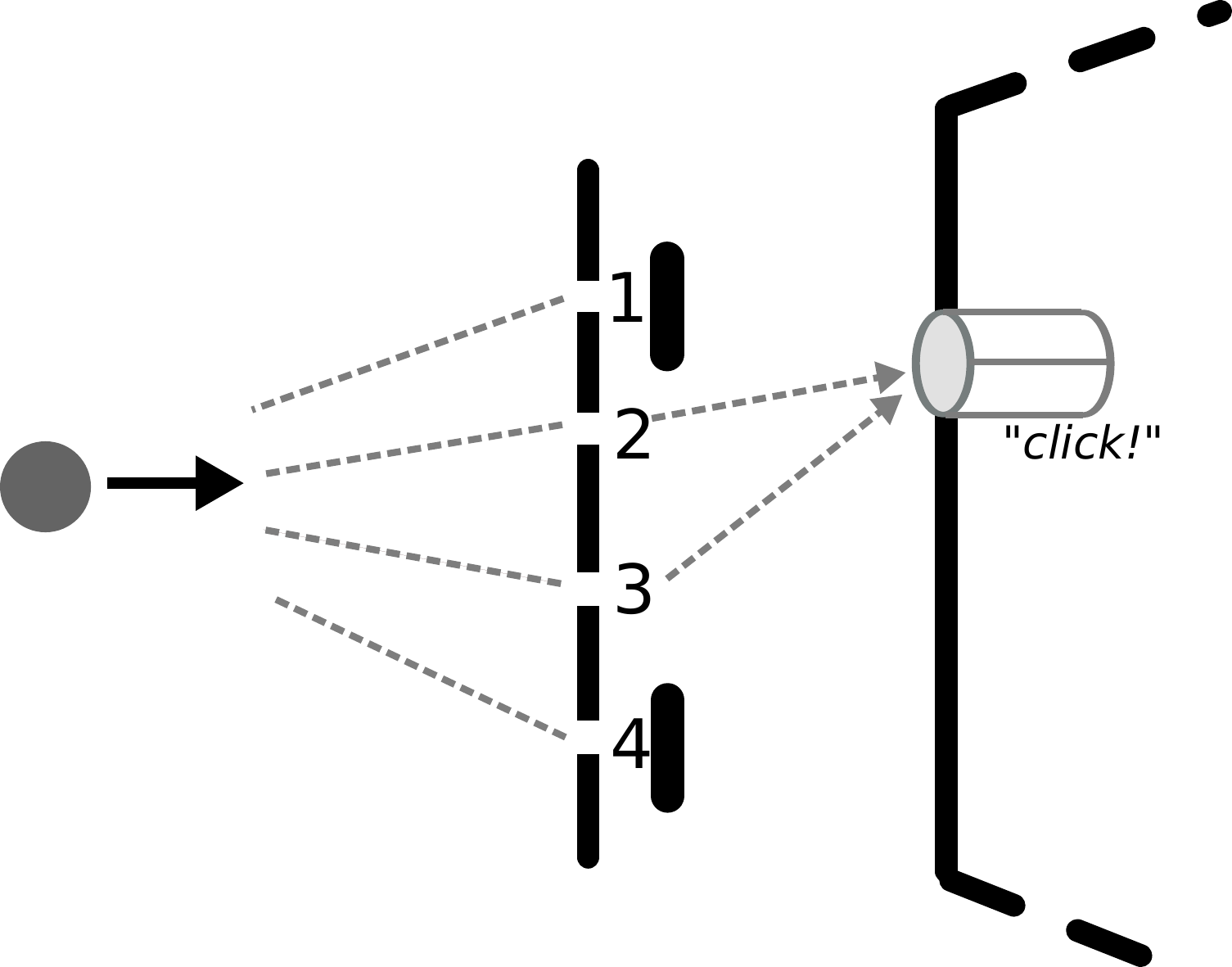}
\caption{\textbf{Higher-order interference.} Consider a particle which can pass one of $M$ (here: $M=4$) slits, where some
of the slits may be blocked by the experimenter (indicated by the black bars). After passing the multi-slit setup,
the particle may trigger a certain event, for example the click of a detector localized in a certain area
of the screen. We are interested in the probability $p_J$ of this event, given that slits $J\subset\{1,2,\ldots,M\}$
are open (for example $p_{23}$ in the depicted setup).\\
Classically, the probability of such an event given that all four slits are open, $p_{1234}$, equals $p_1+p_2+p_3+p_4$, where $p_i$ is
the probability assuming than only slit $i$ is open. This is violated in quantum theory due to interference.
However, even in quantum theory, the total probability can be computed from contributions
of \emph{pairs of slits} only: we have $p_{1234}=p_{12}+p_{13}+p_{14}+p_{23}+p_{24}+p_{34}-2p_1-2p_2-2p_3-2p_4$.
It is in this sense that quantum theory has second-, but no third- or higher-order interference. The definition of
interference that we use is not restricted to spatially arranged slits, but is formulated generally for any set of $M$ perfectly distinguishable
alternatives in a probabilistic theory.
}
\label{fig_interference}
\end{center}
\end{figure}

Before discussing their physical interpretation and motivation in more
detail, we point out that all of our postulates refer to \emph{single}
systems only. This is in contrast to earlier reconstructions of
quantum theory~\cite{Hardy01, DakicBrukner, MMAxioms, GiulioDer11}
which rely heavily on properties of composite systems. Our motivation
to rely on single systems is as follows. It is not clear that the
notion of subsystems and their composition, as it is often used in
information-theoretic circuit diagrams and category-theoretic
considerations, applies to physics without change in its full operational interpretation.
For example, if a composite quantum system consists of spacelike separated subsystems,
then the causal spacetime structure of special relativity imposes additional
complications when describing the possible joint measurements on the composite system~\cite{SorkinImpossibleMeasurements}.
These additional restrictions are usually not captured by operational approaches, which just declare a set of states
and measurements for the composite system, and postulate that these can in principle be implemented to arbitrary accuracy.
Therefore, a safe strategy for an operational approach seems to be to avoid making assumptions about the state space structure of composite
systems, and to talk only about stand-alone systems. These may or may not correspond to effective physical subsystems that can be controlled by an agent in a laboratory.

Moreover, there has recently been a surge of interest in finding
compelling physical principles that explain the specific contextuality
behavior of quantum theory as compared to other probabilistic theories.
This  line of research aims at analyzing the \emph{single-system analogue}
of quantum non-locality, and understanding its specific characteristics
in terms of principles such as ``consistent exclusivity''~\cite{Fritz}. Our results also
contribute to this line of research by showing that Postulates 1 and 2
are sufficient to guarantee that systems satisfy consistent exclusivity.

We do not claim that our postulates are the only reasonable ones, but we
think that they -- like other recent reconstructions -- are more natural
than the usual abstract formulations which simply presume Hilbert spaces, complex
numbers, and operators. Moreover, as we discuss below, we think that
our formulation is especially suitable for the search for interesting
and physically reasonable modifications of quantum theory; that is,
state spaces that are not described by the Hilbert space formalism
but are otherwise consistent and physically plausible.

Comparison to other reconstructions can help uncover logical relations
between various physical structures of our world. For example, our fourth
postulate (observability of energy) is used to rule out non-complex
Hilbert spaces in this work, while in other reconstructions this role is
usually played by the the postulate of \emph{tomographic locality},
which states that joint states on composite systems are uniquely determined
by local measurement statistics and their correlations. Thus, one may
argue that there is a logical relationship between
tomographic locality and observability of energy, and thus ultimately
with the fact that we observe Hamiltonian mechanics in our world.

We will now give a short discussion of the interpretation of our postulates.
To clarify the terms in Postulate 1, a set of states is perfectly
distinguishable if there is a measurement whose outcomes can be
paired one-to-one with the states so that each measurement outcome has probability
one when its corresponding state has been
prepared, and probability zero when any of the
other states have been prepared.  A state of maximal knowledge (a
``pure state'') is a state $\omega$ which cannot be written as a
nontrivial convex combination of states, i.e. as $\omega = p \sigma
+ q \tau$ where $p + q = 1$, $p, q > 0$, and $\sigma \ne \tau$.
That is, it cannot be viewed as arising from a lack of knowledge
about which of two distinct states has been prepared.

Postulate 1 can be viewed as a generalization of the spectral
decomposition of every quantum density matrix as a convex combination
of orthogonal rank-one projectors onto orthogonal eigenstates of the
density matrix. 
However, our postulate is stated purely in terms of the convex
structures of the set of states and of measurement outcomes; the notion of spectrum of an operator is not involved.
An important part of the physical significance of this postulate is that
it appears likely to be needed for an information theory and probably a
statistical mechanics that share desirable and physically fundamental
properties with those supported by quantum theory.  In particular, it is a plausible conjecture that this postulate
implies the correspondence of two natural ways of defining entropies for states in generalized probabilistic theories~\cite{Entropy,SWentropy}: the first as the minimal entropy
of the outcomes of a fine-grained measurement made on the state, and
the second as the minimal entropy of a preparation of the state as a
mixture of pure states.

Postulate 2 expresses a fundamental symmetry:
given any integer $n$, all $n$-level systems are informationally equivalent.
That is, we can transmit (not necessarily copy) the state of any $n$-level
system to any other system without losing information, at least in principle.
This implies a certain minimal amount of possible reversible dynamics or computational power.

Postulate 3, that the system exhibits at most ``second-order
interference,'' is based on the notion of multi-slit interference
introduced by Rafael Sorkin~\cite{SorkinMeasure94}. This is a
manifestly physical assumption which is currently under experimental
investigation~\cite{Sinha2010, SollnerWeihs12}. The precise notion of an
interference experiment will be defined in Section~\ref{SecInterference} below; an
illustration is given in Figure~\ref{fig_interference}.

This postulate suggests a possible route towards obtaining concrete predictions
for conceivable third-order interference in experiments:
drop the third postulate, and work out the new set of theories that
satisfy only Postulates 1 and 2 (and possibly 4). As we will show, any
system of this kind -- if it exists -- has a set of ``filtering''
operations that represent an orthomodular lattice known from quantum
logic~\cite{BirkvN36}, but these filters do not necessarily
preserve the purity of states as they do in quantum theory (equivalently,
the lattice does not satisfy the ``covering law'').  However, these
systems still satisfy the principle of ``Consistent
Exclusivity''~\cite{Fritz}, bringing their contextuality behavior
close to quantum theory, despite the appearance of (non-quantum)
third-order interference.

In this way, our
results hint at possible physical properties of conceivable
alternative theories against which quantum theory can be tested in
interference experiments, and which may be of independent
mathematical interest. In particular, the existence of theories exhibiting
higher-order interference and containing quantum theory as a subtheory
has been conjectured for several years. Preliminary results indicate
interesting physical properties of those theories~\cite{DakicPaterekBrukner},
but the concrete construction of the corresponding state spaces is still an
open problem. We hope that our approach can help to make progress on this question.

We obtain our main result by first showing that the first three
postulates bring us very close to quantum theory: they imply that
systems are described by finite-dimensional irreducible (simple) formally real
Jordan algebras, or are classical.  Moreover, these
three postulates precisely characterize this class of theories, since
classical systems and irreducible Jordan algebras all satisfy
Postulates 1-3.  As Jordan, von Neumann and Wigner~\cite{JvNW1934}
showed, the formally real irreducible Jordan algebras are the real,
complex, and quaternionic quantum theories (for all finite
dimensions), one exceptional case (the $3 \times 3$ octonionic
``density matrices'') and the spin factors (ball-shaped state spaces)
of all finite dimensions.  Standard complex quantum theory is the only
one among these which also satisfies the fourth postulate.

The association of energy with a conserved physical quantity is an important
principle of both quantum and classical theory, exhibited for example in  
the Lagrangian formulation of classical mechanics in the guise of Noether's 
theorem;  this provides some motivation for our energy observability postulate.

Further, Postulates 1, 2 and 4 seem likely to be necessary---or at
least sufficient---to run standard statistical mechanics
arguments, a possibility we will explore in further work. We have
already mentioned the conjecture that Postulate 1 implies the equivalence of
measurement and preparation entropy, which likely has relevance to
thermodynamic processes and Maxwell's demon arguments.  Reversible
processes, the subject of Postulate 2, are even more crucial in
classical and quantum thermodynamics.

\section{Operational Probabilistic Theories}

In this section, we summarize the standard mathematical framework for operational probabilistic theories, and
give needed definitions and facts about convexity and
cones.  References for the mathematics include~\cite{Webster}
and~\cite{HilgertHoffmanLawson}.  More details on the framework can be found in e.g.
\cite{BBLW07broadcast},
\cite{MMAxioms}, \cite{BBLW08clone}, \cite{BBLW08tele}, \cite{CozThesis}; also, \cite{JanottaHinrichsen,Pfister} offer accessible introductions.
This review is primarily to fix notation and clarify the specific version used here.

The primitive elements of operational probabilistic theories are experimental devices and
probabilities. In particular, experimental devices can be classified
into preparations, transformations, and measurements.  With each
use, a preparation device (such as an oven, antenna, or laser)
outputs an instance of a physical system, denoted by $A$, in some
state $\omega$ specified by the type of device and its various
settings.  The system then passes through a transformation device
(such as a beam splitter, or Stern-Gerlach magnet) which modifies
the state of the system, in a potentially non-deterministic fashion.
Finally, a measurement device takes in the system, and one of a
distinct set of outputs (such as a light flashing, or a pointer
being in some range of possible positions) signals the measurement
outcome.
Even though we motivate the formalism by example
of such laboratory devices, the resulting operational framework is not restricted to this setting and may also be used to
describe other physical processes.

A main purpose of a physical theory in this framework is
to specify the probabilities of the outcomes of any measurement made
on a system that has been prepared in a given state.  To this end,
single measurement outcomes, called \emph{effects}, will be denoted
by lowercase letters such as $e$.  The probability of obtaining an
outcome $e$, given state $\omega$, will be denoted $e(\omega)$.

By standard arguments, each state can be specified by a minimal list
of measurement outcome probabilities, which contains sufficient
information to predict the probabilities of \emph{all} measurements
that can be in principle performed on the system.  Using this idea and
a further convexity argument, states can therefore be represented as
elements of a real linear space of some finite dimension $K_A$, which
we denote also by $A$.  Further, for each system $A$ there is a convex
compact subset, $\Omega_A \subset A$, of \emph{normalized states} in a
real affine space of dimension $K_A-1$ which is embedded in $A$ as an
affine plane not intersecting the origin.  The nonnegative multiples
of elements of $\Omega_A$ form a 
cone $A_+ \subset A$, of \emph{unnormalized states}.
This cone has several useful properties: first, it is
topologically \emph{closed}; second, it has \emph{full dimension},
i.e.\ its linear span is all of $A$; and third, it is \emph{pointed}, which means
that the only linear subspace it contains is $\{0\}$. Cones with these three
properties are also called \emph{regular}.

Effects then become linear functionals from $A$ to $\R$ such that
$0 \le e(\omega) \le 1$ for all $\omega \in \Omega_A$, i.e.\ they give valid
probabilities on normalized states.
As linear functionals from the vector space $A$ to 
the field $\R$ over which it is defined, effects
are elements of the dual space $A^*$, which is the 
vector space of all such functionals.
The nonnegative multiples of
effects constitute the dual cone $A_+^* := \{ e \in A^*: \forall
\gamma \in A_+ ~e(\gamma) \ge 0 \}$.  Given our embedding of
$\Omega_A$ in $A$, there is a unique \emph{unit} functional $u_A \in
A^*$ that evaluates to $1$ on every element of $\Omega_A$.  The set of all effects is the
\emph{unit order interval}, $[0, u_A] := \{e \in A^*_+: 0 \le e \le u_A\} \subset A^*_+$. This notation uses the ordering
obtained from the regular cone $A_+^*$, writing $x\leq y$ for $y-x\in A_+^*$.

For a given system, not all mathematically valid effects may be ``operationally possible'' measurement outcomes, so we define a
\emph{subset} $\ce$ of the full set of effects $[0, u_A]$, which we
call the \emph{allowed effects}.  Thus we are not making
the assumption sometimes called the  
``no-restriction hypothesis''~\cite{ChiriEtAl2010, MMAxioms, JanottaLal}
or ``local saturation''~\cite{BarnumWilceCategories2009}, nor the equivalent dual 
requirement (discussed, e.g., in~\cite{Barnum2003}, where it is considered as a kind
of analogue, for effect algebras, of Gleason's theorem) 
that the set of states be the full set of mathematically consistent states on the
set of effects.  The reader should bear in mind
that some authors use just ``effects'' to refer to what we 
call ``allowed effects'', and say something like ``mathematically consistent effects'' 
to refer to what we are just calling effects.
We make weak, operationally natural
assumptions on the subset $\ce$: it is convex and topologically
closed, contains $u_A$, and for every $x \in \ce$, $u_A - x$ is also
in $\ce$ (so that $x$ can be part of at least one complete
measurement, namely $\{x, u_A - x\})$. We also assume that $\ce$ has full
dimension (otherwise, there would be states $\varphi\neq \omega$ that give
the same outcome probabilities for all allowed measurements, which means that
we would not have called them ``different states'' to start with).

We define a \emph{measurement} as any collection of allowed effects $e_i$ such
that $\sum_i e_i=u_A$.\footnote{It is possible to imagine physical situations
where there are further restrictions on which effects can occur together in an
actual measurement; to model these situations, one would have to use an
even more general mathematical framework. We are not considering such theories here.}
Since we can imagine post-processing the
output of such a measurement such that a chosen pair $e_i$ and $e_j$
of outcomes are grouped together as a single outcome (a
``coarse-graining'' of the measurement), we also assume that $e_i +
e_j$ is allowed.  In brief, we assume that whenever $e_i, e_j$
are allowed effects with $e_i + e_j \le u_A$, $e_i + e_j$ is allowed.
From our assumptions, it follows that the set of allowed effects is the unit order
interval $[0,u_A]$ in a regular subcone $A_+^\sharp$ (containing
$u_A$) of the dual cone.
If $A_+^\sharp = A_+^*$, we say that \emph{all effects are allowed}; 
in our framework,
this is equivalent to the ``no-restriction hypothesis'', or ``local saturation'', 
mentioned above.

We will need the notion, standard in
  linear algebra, of the \emph{dual} (sometimes called \emph{adjoint})
  $T^*$ of a linear map $T: A \rightarrow A$.  This is the
  linear map $T^*: A^* \rightarrow A^*$ defined by the condition
  $( f, Tx) = ( T^*f , x )$, where $(
  . , . ): A^* \times A \rightarrow \R$ is the canonical ``dual
  pairing'' of $A^*$ and $A$, sometimes called the ``evaluation map'':
  $( f , x ) := f(x)$.

Associated with every system there is also a set of allowed 
\emph{transformations}, which are linear maps
$T: A \rightarrow A$, taking states to states, i.e.\ satisfying $T(A_+) \subseteq A_+$ (a
property called \emph{positivity}).  Transformations are required to be
\emph{normalization-nonincreasing}, i.e.\ $u_A(T(\omega))\leq 1$
for all $\omega\in\Omega_A$.  The set of allowed transformations is also closed topologically and under composition.
If all effects are allowed, it follows
from positivity and normalization that $e \circ T \in \ce$ for all
allowed effects $e$ (all elements of $\ce$); otherwise we explicitly
require this (i.e., that $T^*(\ce) \subseteq \ce$). 
  Since $\ce$ is the unit order interval in $A_+^{\sharp}$, it is equivalent
  (for normalization-nonincreasing $T$) to require that $T^*(A_+^\sharp)
\subseteq A_+^\sharp$.  We note also that the normalization-nonincrease
  condition is equivalent to the dual condition $T^*(u_A) \le u_A$.  
An allowed transformation $T$ is called \emph{reversible} if its
inverse $T^{-1}$ exists and is also an allowed transformation. It
follows that reversible transformations $T$ preserve normalization:
$u_A(T(\omega)) = u_A(\omega)$ for all $\omega \in A_+$ (though
these are not in general the only normalization-preserving
transformations). The set of all reversible transformations on a system $A$ is a compact group $\G_A$ with Lie algebra $\mathfrak{g}_A$.
For a transformation $T$, the number
$u_A(T(\omega))$ can be interpreted as the probability of
transformation $T$ occurring, if a system prepared in state $\omega$
is subjected to a process that has as a possible outcome the
occurrence of $T$. In other words, transformations can be part of
an \emph{instrument} in the sense of~\cite{DaviesLewis}.

A system described by standard complex $n$-dimensional quantum theory
fits into this framework.  Its ambient real vector space $A$ is the
$n^2$-dimensional space of complex Hermitian $n \times n$-matrices, the cone of states $A_+$ is
the set of positive semidefinite matrices, $\Omega_A$ is the set of density matrices
(the intersection of $A_+$ with the affine plane $\{\rho: \tr \rho = 1\}$), the order unit is the functional
$\mathbf{1}:\rho \mapsto \tr \rho$, and the allowed effects are the unit order interval in the
dual cone, i.e., the functionals $\rho \mapsto \tr (E \rho)$ where $0
\le E \le \mathbf{1}$.  The allowed transformations are the trace-nonincreasing completely positive
maps $A \rightarrow A$, and the reversible transformations are the
maps $\rho \mapsto U \rho U^\dagger$ for unitary matrices $U$.

We now describe some further important notions and facts about this
type of theory and the relevant mathematical structures that will be
used in our discussion.

A cone $A_+$ is \emph{reducible} if the ambient space decomposes into
two nontrivial subspaces such that every extremal ray of the cone lies
in one or the other of these subspaces.    A system is called reducible
if its cone of unnormalized states is reducible.  Intuitively, information about which of
these two summands the state is in, is classical information.  Every
cone in finite dimension has a decomposition as a finite sum
$\oplus_{i=1}^n A_i$ of irreducible cones, and if these irreducible
components are all one-dimensional any base for the cone is affinely
isomorphic to the simplex of probability measures over $n$ outcomes,
so we say the system is \emph{classical}.  Its faces are the
subsimplices generated by the subsets of outcomes, its reversible
transformations are the permutations of the vertices, and more general
transformations are given by substochastic matrices. 

One can identify $A^*$ with $A$ by introducing an inner product
$\langle . , . \rangle$ on $A$, and interpreting the inner product as
functional evaluation: $e(\omega) = \langle e , \omega \rangle$.
 Via this isomorphism the dual cone $A_+^*$ is identified
  with the ``internal dual cone'' relative to the given inner product,
  $A_+^{*int} := \{ y \in A: \forall x \in A_+ \, \langle y , x \rangle
  \ge 0 \}$.  Often, such an inner-product-space formulation is used 
as the basic framework for presenting probabilistic systems and
theories; see for example \cite{Hardy01, Barrett}.  
If an inner product can be introduced in such a way that
$A_+^{*int} = A_+$, the cone is said to be \emph{self-dual} and the
inner product \emph{self-dualizing}; a cone in an inner product space is said
to be \emph{manifestly self-dual} if the inner product is one that
identifies the cone with its dual.

A set of states $\omega_1,\ldots,\omega_n\in\Omega_A$ is called
\emph{perfectly distinguishable} if there are allowed effects
$e_1,\ldots,e_n\in A^\sharp_+$ which can appear in a common
measurement, i.e.\ $e_1+\ldots+e_n\leq u_A$, such that
$e_i(\omega_j)=\delta_{ij}$, that is, $1$ if $i=j$ and $0$ otherwise\footnote{It is equivalent to demand that
$e_1+\ldots+e_n=u_A$, because we can always redefine $e'_1:=e_1,\ldots,e'_{n-1}:=e_{n-1},e'_n:=u_A-\sum_{i=1}^{n-1}e_i$.
}.

A \emph{face} $F$ of a convex set $C$ is a convex subset of $C$ such that
$\alpha \in F$ and $\alpha = \sum_i \lambda_i \omega_i$, $\omega_i \in
C, \lambda_i > 0, \sum_i \lambda_i=1$ implies that all $\omega_i \in
F$.  In other words $F$ is closed under inclusion of anything that can
appear in a convex decomposition of an element of $F$.  
An \emph{exposed face} of a convex set is the intersection of a supporting hyperplane with the set, easily seen to be a face.  

The faces of $A_+$ and those of $\Omega_A$ are in 1-1 correspondence:
the face of $A_+$ corresponding to face $F$ of $\Omega_A$ is just
$\{\lambda \omega : \omega \in F,  \lambda \ge 0\}$.
The relation ``is a face of'' is
transitive: If $G$ is a face of $C$, and $F$ is a face of $G$, then
$F$ is a face of $C$.  The orderings of the set of faces and of the
set of exposed faces by subset inclusion each form a lattice, with
greatest lower bound $F \meet G = F \intersect G$, and least upper
bound $F \join G$, which is the smallest face containing both $F$ and
$G$.  The \emph{face generated by} a subset $S$ of a convex set is the
smallest face containing $S$.  If a lattice has an upper bound, this
is conventionally called $1$, and a lower bound is called $0$; for
$\Omega_A$ we have $1 = \Omega_A$ and $0 = \emptyset$, while for $A_+$, $1
= A_+$ and $0 = \{0\}$, where $0$ is the $0$ of the vector space
$A$. (We adopt the convention that the empty set $\emptyset$ is not
counted as a face of $A_+$.)  An \emph{atom} is a minimal non-zero
element of the lattice; the atoms of the face lattice of a regular
finite-dimensional cone are the extremal rays, $\ray(\omega) := \{
\lambda \omega: \lambda \ge 0\}$ for $\omega$ extremal in $\Omega_A$.
An element of $A_+$ may be called \emph{ray-extremal} if it is a
nonnegative multiple of a pure state of $\Omega_A$.

Quantum systems are self-dual, with all effects allowed, and with the self-dualizing
inner product usually chosen to be $\langle X, Y \rangle = \tr (XY)$.
(For this reason, the dual cone is often identified with the positive
semidefinite operators, and the effects with operators $E$ such that
$0 \le E \le \mathbf{1}$, rather than with the functionals $\rho \mapsto \tr E
\rho$ associated with such operators.)  The faces of a quantum system,
which are all exposed, correspond to the subspaces $S$ of the
underlying Hilbert space: the face $F_S$ of $\Omega$ corresponding to
such a subspace $S$ consists of the density matrices $\rho$ whose
images, when viewed as linear operators on that Hilbert space, are contained in $S$.
Equivalently, they are those density matrices whose convex
decompositions into rank-one projectors involve nonzero probabilities
only for projectors onto subspaces of $S$.

\section{Consequences of Postulates 1+2}

We call a list of $n$ perfectly distinguishable pure states
a \emph{frame} of size $n$, or $n$-frame.  The convex hull of such a
set of states is a simplex, isomorphic to the space of probability
measures on $n$ alternatives, which we call a ``classical subspace''
of the state space.  For every finite-dimensional system $A$, there is a
largest frame size $N_A$; frames of this size are called \emph{maximal}.
In quantum theory, a frame corresponds to a set of mutually
orthogonal pure states, and it is maximal if the corresponding state vectors
are an orthonormal basis of the underlying Hilbert space.

Using the concepts we have introduced, our first two postulates can be stated as follows:

\begin{postulate}
Every state $\omega \in \Omega$ has a decomposition of the form $\omega = \sum_i p_i \omega_i$, for some
probabilities $p_i\geq 0$, $\sum_i p_i=1$, and some $n$-frame $\omega_1,\ldots,\omega_n$, for some $n\in\N$.
\end{postulate}
\begin{postulate}
If $\omega_1,\ldots,\omega_n$ and $\varphi_1,\ldots,\varphi_n$ are $n$-frames for some $n\in\N$,
then there is a reversible transformation $T$ such that $T\omega_i=\varphi_i$ for all $i$.
\end{postulate}

We could paraphrase Postulate 1 as ``every state lies in some classical
subspace'', and Postulate 2 as ``all classical subspaces of a given
size are equivalent''.

\begin{proposition}
\label{SpectralityImpliesSaturation}
Postulates 1 and 2 imply that all effects are allowed.
\end{proposition}
\begin{proof}
We show that every effect $e \in A_+^*$ that generates an exposed
ray of $A_+^*$ is allowed,
i.e.\ an element of $A_+^\sharp$.  It follows that all
effects are allowed, since the exposed
rays generate $A_+^*$ via convex combinations and closure.

Thus, let $e\in A_+^*$ be an effect with $\max_{\omega\in\Omega_A} e(\omega)=1$ such that the set of non-negative
multiples of $e$ is an exposed ray of $A_+^*$.
By the definition of exposed ray, there is an $x \in A_+$ such that every effect $f\in A_+^*$ with $f(x)=0$
must be a non-negative multiple of $e$; consequently, if $f(x)=0$, $f\in A_+^*$ and $\max_{\omega\in\Omega_A}f(\omega)=1$
then $f=e$. We may choose $x$ to be normalized.

According to Postulate 1, there is some $n\in\N$ and some frame $\omega_1,\ldots,\omega_n$ such that $x=\sum_{j=1}^n \lambda_j \omega_j$;
we may choose the $\lambda_j$ to be non-zero. The corresponding effects will be denoted $e_1,\ldots,e_n$, i.e.\ $e_i(\omega_j)=\delta_{ij}$.
Since $e(x)=0$ we have $e(\omega_j)=0$ for all $j=1,\ldots,n$.

We define the \emph{maximally mixed state} $\mu$ by integrating with Haar measure over the group of reversible transformations; that is,
choose any pure state $\omega$, and set $\mu:=\int_{\G_A} G\omega\, dG$. This state also has a frame decomposition $\mu=\sum_{i=1}^N \eta_i \varphi_i$
with $N\in\N$, $\eta_i>0$, and $\varphi_1,\ldots,\varphi_N$ a frame with corresponding effects $f_1,\ldots,f_N$ such that $f_i(\varphi_j)=\delta_{ij}$.

According to Postulate 2, there is a reversible transformation $T\in\G_A$ such that $T\varphi_i=\omega_i$ for all $i=1,\ldots,\min\{n,N\}$.
Suppose that $n\geq N$, then
\[
   \mu=T\mu=\sum_{i=1}^N \eta_i T\varphi_i = \sum_{i=1}^N \eta_i \omega_i,
\]
hence $e(\mu)=0=\int_{\G_A} e(G\omega)\, dG$. Since $G\mapsto e(G\omega)$ is a continuous non-negative function on $\G_A$, we
must have $e(G\omega)=0$ for all $G\in\G_A$, and thus $e(\omega')=0$ for all pure states $\omega'$. Since the pure states span the
full linear space, we obtain $e=0$, which is a contradiction.

Thus we have $n<N$. Consider the allowed effect $f_N\circ T^{-1}$. It satisfies
\[
   f_N\circ T^{-1}(x)=\sum_{j=1}^n \lambda_j f_N(T^{-1}\omega_j)=\sum_{j=1}^n \lambda_j f_N(\varphi_j)=0,
\]
and since $\max_{\omega\in\Omega_A}f_N\circ T^{-1}(\omega)=1$, we have $f_N\circ T^{-1}=e$; in particular,
$e$ is an allowed effect.
\end{proof}

For the following proposition, recall that a set of states is said to \emph{generate} a face
$F$ if $F$ is the smallest face that contains these states.

\begin{proposition}
\label{FramesAndFaces}
Postulates 1 and 2 imply that every face of $\Omega$ is generated
by a frame.  Any two frames that generate the same face $F$ have the same
size, called the \emph{rank} of $F$, and denoted $|F|$. Moreover, if 
$G\subsetneq F$ then $|G|<|F|$,
and every frame of size $|F|$ in $F$ generates $F$.
\end{proposition}

\begin{proof}
A face is generated by any element of its relative
interior. By Postulate 1, such an element is in the convex hull of a frame; this frame also generates
the face.

Let $F$ be any face, and suppose there are two frames $\varphi_1,\ldots,\varphi_m$ and $\omega_1,\ldots,\omega_n$
with $m<n$ that both generate $F$, and $e_1,\ldots,e_n$ effects such that $e_i(\omega_j)=\delta_{ij}$ and $\sum_i e_i\leq u$.
Let $F'$ be the face generated by $\omega_1,\ldots,\omega_m$, then
$G:=\{x\in\Omega\,\,|\,\, e_n(x)=0\}$ is a face of $\Omega$ containing $F'$ but not containing $\omega_n$, so $F'\subsetneq F$. Due to Postulate 2, there is
a reversible transformation $T$ with $T\varphi_i=\omega_i$ for $i=1,\ldots,m$, so $TF\subseteq F' \subsetneq F$. Since $TF$ is a proper face of $F$,
it must have smaller dimension, which contradicts the invertibility and thus reversibility of $T$.
Similarly, if we had $G\subsetneq F$ and $|G|\geq |F|$, then a reversible transformation could map $F$ into $G$, which is a contradiction, too.

If $\omega_1,\ldots,\omega_{|F|}$ is any frame on $F$, and $G$ the face that it generates, then $G\subseteq F$, and some reversible transformation $T$ will map it to some
other frame of the same size that generates $F$. Hence $TG=F$, and this contradicts $G\subsetneq F$.
\end{proof}

\begin{proposition} \label{proposition on self duality}
Postulates 1 and 2 imply that $A_+$ is self-dual, with a corresponding self-dualizing inner product that satisfies
$\langle T\varphi,T\omega\rangle=\langle \varphi,\omega\rangle$ for all reversible transformations $T$, i.e.\ such that
all reversible transformations are orthogonal.
The inner product can be chosen so that the corresponding norm $\|\omega\|:=\sqrt{\langle\omega,\omega\rangle}$
attains the value $1$ on all pure states, and is strictly less than $1$ for all mixed states.
\end{proposition}
\begin{proof}
Ref.~\cite{MuellerUdudec} shows that bit symmetry and the fact that all effects are allowed imply this proposition.
Bit symmetry is the $2$-frame case of Postulate 2, and we have shown that all effects are allowed in Proposition~\ref{SpectralityImpliesSaturation}.
\end{proof}

Henceforth, except when we explicitly state otherwise, we
identify $A^*$ with $A$ via an inner product satisfying the conditions in the above proposition.
Since reversible transformations $T$ are normalized, we have $T^*(u_A)=u_A$.
Moreover, $T^*=T^{-1}$ by orthogonality.  $T^*$ is also a reversible transformation; thus, if we regard
$u_A$ now as an element of $A$, we obtain that $T^{-1}u_A=u_A$ for all $T^{-1}$. This proves the following:

\begin{proposition}
Postulates 1 and 2 imply that $u_A$ is invariant under all reversible transformations.
\end{proposition}

\begin{proposition}\label{PropFrameExtension}
Postulates 1 and 2 imply that every frame $\omega_1,\ldots,\omega_n$ can be
extended to a frame $\omega_1,\ldots,\omega_n,\ldots,\omega_N$ which generates $A_+$,
i.e.\ $N=|A_+|$.
\end{proposition}
\begin{proof}
Let $\varphi_1,\ldots,\varphi_N$ be any frame that generates all of
$A_+$, with effects $e_1,\ldots,e_N$ such that
$e_j(\varphi_i)=\delta_{ij}$ and $\sum_j e_j=u_A$.  Then
$\varphi_1,\ldots,\varphi_n$ is itself a frame of size $n$; thus,
according to Postulate 2, there is a reversible transformation $T$
with $T\varphi_i=\omega_i$ for $i=1,\ldots,n$. For $i>n$, define
$\omega_i:=T\varphi_i$. Set $e'_j:=e_j\circ T^{-1}$, then
$e'_j(\omega_i)=\delta_{ij}$ and $\sum_j e'_j = u_A$, and so we have
extended $\omega_1,\ldots,\omega_n$ to a frame with $N$ elements.
\end{proof}

The following proposition will turn out to be useful in several proofs.
\begin{proposition}
\label{PropTakeFramesTogether}
Postulates 1 and 2 imply that if $\omega_1,\ldots,\omega_n$ are mutually orthogonal pure states,
then they are a frame, and $\sum_{i=1}^n \omega_i\leq u_A$.
\end{proposition}
\begin{proof}
We have to find effects $e_1,\ldots,e_n$ with
$e_i(\omega_j)=\delta_{ij}$ and $\sum_{i=1}^n e_i\leq u_A$.  To this
end, we will first construct a decomposition of the order unit. By
self-duality and Proposition~\ref{proposition on self duality} , $\varphi:=u_A/\langle u_A,u_A\rangle$ is a state in
$\Omega$, hence there is a frame $\varphi_1,\ldots,\varphi_N$ with
$N=|A_+|$ and $\lambda_i\geq 0$ such that $u_A=\langle
u_A,u_A\rangle\varphi=\|u_A\|^2 \sum_{i=1}^N \lambda_i \varphi_i$. For
any permutation $\pi:\{1,\ldots,N\}\to\{1,\ldots,N\}$, the states
$\varphi_{\pi(1)},\ldots,\varphi_{\pi(N)}$ are again a frame; thus,
there is a reversible transformation $T_\pi$ with $T_\pi \varphi_i =
\varphi_{\pi(i)}$. Hence (using the invariance of $u_A$ under reversible
transformations)
\[
   u_A=T_\pi u_A = \|u_A\|^2 \sum_{i=1}^N \lambda_i \varphi_{\pi(i)}=\|u_A\|^2 \sum_{i=1}^N \lambda_i \varphi_i.
\]
   Taking the  inner  product with  $\varphi_j$ shows  that
$\lambda_{\pi^{-1}(j)}=\lambda_j$;   since  this   is  true   for  all
permutations, all $\lambda_j$ are  equal to some $\lambda>0$. Finally,
$1=\langle    u_A,\varphi_1\rangle=\|u_A\|^2    \lambda$,    and    so
$u_A=\sum_{i=1}^N  \varphi_i$.  If  $\omega_1,\ldots,\omega_N$ is  any
other frame  of size  $N$, then  Postulate 2 implies  that there  is a
reversible transformation $T$ such that $T\varphi_i = \omega_i$, hence
$u_A=T u_A = T\sum_{i=1}^N  \varphi_i = \sum_{i=1}^N \omega_i$.  Thus,
we  have shown that  \emph{every maximal  frame adds  up to  the order
  unit}.

Now we show the statement of the proposition by induction on $n$. Start with $n=1$. Any pure state $\omega_1$ is by definition a frame of size $1$.
Moreover, if $\varphi\in\Omega$, then the Cauchy-Schwarz inequality yields
\[
   \langle \omega_1,\varphi\rangle\leq \|\omega_1\|\cdot\|\varphi\|\leq 1,
\]
hence $\omega_1\leq u_A$. Now suppose the statement of the proposition is true for some $n$, and consider pure mutually orthogonal states
$\omega_1,\ldots,\omega_{n+1}$. Set $e_1:=\omega_1,\ldots,e_n:=\omega_n$, and $e_{n+1}:=u_A-\sum_{i=1}^n e_i$. By the induction hypothesis,
$e_{n+1}\geq 0$, and so $e_1,\ldots, e_{n+1}$ is a measurement with $e_i(\omega_j)=\delta_{ij}$ for $1\leq i,j\leq n+1$. Thus, $\omega_1,\ldots,\omega_{n+1}$
is a frame. According to Proposition~\ref{PropFrameExtension}, it can be extended to a maximal frame $\omega_1,\ldots,\omega_N$, and then
$\sum_{i=1}^N \omega_i = u_A$ shows that $\sum_{i=1}^{n+1}\omega_i \leq u_A$.
\end{proof}

Recall that for any subset $S$ of an inner product space $V$ its \emph{orthogonal complement}
$S^\perp$ is defined by $S^\perp := \{x \in V: \forall y \in S ~\langle x , y \rangle = 0 \}$.

\begin{proposition}
\label{PropFPrime}
Postulates 1 and 2 imply that for every face $F$ of $A_+$, the set $F':=F^\perp\cap A_+$ is a face of $A_+$
of rank $|F'|=N-|F|$, where $N=|A_+|$, and we
have $(F')'=F$. Furthermore, if $\varphi_1,\ldots\varphi_n$ is any frame that is contained in some face $F$,
then it can be extended to a frame $\varphi_1,\ldots,\varphi_n,\ldots,\varphi_{|F|}$ that generates $F$.
\end{proposition}
\begin{proof}
Let $\omega\in F'$ be any element, and $0<\lambda<1$, $\omega_1,\omega_2\in A_+$ such that $\omega=\lambda\omega_1+(1-\lambda)\omega_2$.
Then, for every $f\in F$, we have $0=\langle f,\omega\rangle=\lambda\langle f,\omega_1\rangle + (1-\lambda)\langle f,\omega_2\rangle$. Due to self-duality,
we have $\langle f,\omega_i\rangle\geq 0$ for $i=1,2$, hence $\langle f,\omega_1\rangle=\langle f,\omega_2\rangle=0$ for all $f\in F$. This shows that
$\omega_1,\omega_2\in F'$, hence $F'$ is a face.

Now we determine the rank of $F'$. Let $\omega_1,\ldots,\omega_{|F|}$ be any frame that generates $F$, and
$\varphi_1,\ldots,\varphi_{|F'|}$ be a frame that generates $F'$. Then $\langle \omega_i,\varphi_j\rangle=0$ for all $i,j$, and so
Proposition~\ref{PropTakeFramesTogether} tells us that
both frames taken together are a frame in $A_+$, proving that $|F|+|F'|\leq N$. Extend $\omega_1,\ldots,\omega_{|F|}$ to a frame on $A_+$,
then $\omega_{|F|+1},\ldots,\omega_N$ are orthogonal to $F$ and thus a frame in $F'$, showing that $|F'|\geq N-|F|$, so $|F'|=N-|F|$, and
the extension is actually a generating frame of $F'$. Consequently, $\omega_1,\ldots,\omega_{|F|}\in(F')'$, and since $|(F')'|=N-|F'|=N-(N-|F|)=|F|$,
these states generate $(F')'$. Since they also generate $F$, we must have $F=(F')'$.

Now suppose that $\varphi_1,\ldots,\varphi_n$ is any frame contained in $F$; let $\omega_1,\ldots,\omega_{|F|}$ be any frame that
generates $F$. According to Proposition~\ref{PropFrameExtension}, we can extend it to a frame $\omega_1,\ldots,\omega_{|F|},\ldots,\omega_N$
that generates all of $A_+$; moreover, the $\omega_i$ with $i\geq |F|+1$ generate $F'$. But then $\langle \omega,\omega_i\rangle=0$ for all $i\geq |F|+1$ and $\omega\in F$. Thus, the set of
states $\varphi_1,\ldots,\varphi_n,\omega_{|F|+1},\ldots,\omega_N$ is a set of mutually orthogonal pure states and thus, due to Proposition~\ref{PropTakeFramesTogether},
a frame. Using Proposition~\ref{PropFrameExtension} again, we can find states $\varphi_{n+1},\ldots,\varphi_{|F|}$ such that
$\varphi_1,\ldots,\varphi_n,\varphi_{n+1},\ldots,\varphi_{|F|},\omega_{|F|+1},\ldots,\omega_N$ is a frame generating $A_+$. For $i\geq |F|+1$ and $j$ arbitrary, we have
$\langle \varphi_j,\omega_i\rangle=0$, and since these $\omega_i$ generate $F'$, we have $\langle \varphi_j,\omega\rangle=0$ for all $\omega\in F'$. Thus $\varphi_j\in(F')'=F$,
and we have extended $\varphi_1,\ldots,\varphi_n$ to a frame generating $F$.
\end{proof}

As mentioned in Section~\ref{SecIntroduction},
Postulates 1 and 2 imply that there is a special transformation called a \emph{filter} associated with
each face of the state space. The next theorem shows that certain projections are positive (recall that a linear map is \emph{positive} if
it maps the cone $A_+$ into itself),
and in Section \ref{SecJordan} we will further show that these projections have the additional properties required of filters.

\begin{theorem}
\label{ThmPosProjection}
Postulates 1 and 2 imply that for every face $F$ of $A_+$, the
orthogonal projection $P_F$ onto the linear span of $F$ is positive.
\end{theorem}

\begin{proof}
Iochum (\cite{IochumThesis}, see also \cite{IochumBook}) has
shown that positivity of all $P_F$ is equivalent to \emph{perfection}.
(For the reader's convenience, and the authors' peace of mind, a proof is included
in Appendix~\ref{AppendixIochum}.)
A cone is called \emph{perfect} if all faces $F$ of $A_+$, regarded as cones
in the linear span $\lin F$, are themselves self-dual with respect to
the inner product inherited from $A$. We will therefore show this property,
establishing the claim.

So let $F$ be any face of $A_+$, and $F^*\subset \lin F$ be the dual
cone with respect to the inner product inherited from $A$. Since $F
\subseteq A_+ = A_+^*$, for $\omega \in F$ we have $\langle
\omega,\varphi\rangle\geq 0$ for all $\varphi\in F$, and so $\omega\in
F^*$.  This proves that $F\subseteq F^*$.  To see the converse
inclusion, let $e$ be any normalized element of $F^*$ (i.e.\ $\langle u_A,e\rangle=1$) that generates
an exposed ray of $F^*$.  This means there exists $\omega\in F$ (which we may
choose normalized) with $\langle e,\omega\rangle=0$ such that $f\in
F^*$ and $\langle f,\omega\rangle=0$ implies $f=\lambda e$ with
$\lambda\in\R$.  But $\omega=\sum_i \lambda_i\omega_i$ for some frame
$\omega_1,\ldots,\omega_k\in F$ and $\lambda_i>0$. Since $\omega$ is
in the face $\{\varphi\in F\,\,|\,\, \langle e,\varphi\rangle=0\}\subsetneq F$, we have $k<|F|$, and extending to a
frame $\omega_1,\ldots,\omega_k,\ldots,\omega_{|F|}$ on $F$ gives
$\omega_{|F|}\in F \subseteq F^*$ as well as $\langle
\omega_{|F|},\omega\rangle=0$, hence $e=\omega_{|F|}\in F$. Since the
exposed rays generate $F^*$, this proves that $F^*\subseteq F$.
\end{proof}

The properties that we have proven so far turn out to give an
interesting structure known from the field of quantum logic, indeed
sometimes taken as a definition of a quantum logic \cite{HardingInHandbook}.  
As noted above, the set of faces ordered by subset
inclusion is a bounded lattice.  However, from Postulates 1 and
2, we recover more of the logical structure of quantum theory:

\begin{theorem}\label{ThmOrthomodular}
Postulates 1 and 2 imply that the lattice of faces of $A_+$ is an
\emph{orthomodular lattice}.
\end{theorem}

Before giving the proof, recall that orthomodularity is the property that
\begin{equation}
   \label{EqOrthomodularity}
   F\subseteq G\enspace \Rightarrow\enspace G=F\vee (G\wedge F').
\end{equation}
Note that in \cite{Barker} it is shown that for
self-dual cones, orthomodularity of the face lattice in the above
sense is equivalent to the property of \emph{perfection} mentioned in
the proof of Theorem~\ref{ThmPosProjection}.  Furthermore,
in~\cite{ASBook} it is shown that orthomodularity of the face lattice,
according to an orthocomplementation which agrees with ours in case
Postulates 1 and 2 hold, follows from a property called \emph{projectivity}.  In the
next section we will define projectivity and establish
that state spaces satisfying Postulates 1 and 2 are projective, giving
us an alternative proof of orthomodularity.  Here, we proceed with
the direct proof.

\begin{proof}
Constructing $F'$ as the face generated by the extension of a frame generating $F$ shows easily that $(F')'=F$ (as already shown
in Proposition~\ref{PropFPrime}), and that $F\subseteq G$ implies
$F'\supseteq G'$, as well as $F\vee F'=\mathbf{1}\equiv A_+$ and $F\wedge F'\equiv \mathbf{0}\equiv \{0\}$. These properties mean that the operation $'$ is an
\emph{orthocomplementation} on the lattice of faces.  It remains to show that
this orthocomplemented lattice satisfies the orthomodular law, Eq. (\ref{EqOrthomodularity}).
To this end, assume $F\subseteq G$, and let $\omega_1,\ldots,\omega_{|F|}$ be a frame on $F$.  Extend this to a frame on $G$, and further extend the result to a frame on $A_+$,
yielding $\omega_1,\ldots,\omega_N$. 
Then $\omega_{|F|+1},\ldots,\omega_{|G|}$ is a frame on $G\cap F'$; if it did not generate $G\cap F'$,
it could be extended in $G\cap F'$, and to this extension we could append $\omega_1,\ldots,\omega_{|F|}$ to obtain a frame of size larger
than $|G|$ in $G$, which is a contradiction. Hence
$H:=G\cap F'$ is generated by $\omega_{|F|+1},\ldots,\omega_{|G|}$. Since $F\vee H$ is the smallest face containing $F$ and $H$, it
is the smallest face containing $\omega_1,\ldots,\omega_{|G|}$, hence equal to $G$.
\end{proof}

Systems that satisfy Postulates 1 and 2 are operationally
close to quantum theory also with respect to their contextuality
behavior: they satisfy the principle of \emph{consistent exclusivity}~\cite{Fritz}, the single-system generalization of the
recently introduced postulate of \emph{local orthogonality}~\cite{FritzLO}. This is also called \emph{Specker's
Principle}~\cite{Cabello}, and comes in slightly different versions, depending on assumptions of the validity of the principle
in situations where one has more than one copy of a state. Here we are interested in the single-system version that is called $\mbox{CE}^1$ in~\cite{Fritz}.

In order to talk about contextuality, we need a notion of ``sharp measurements'': the analogs of projective measurements in quantum theory.
Following~\cite{CabelloSeveriniWinter}, we call an effect $0\leq e\leq u_A$ \emph{sharp} if it can be written as a sum of normalized ray-extremal effects; that is,
if there are pure states $\omega_1,\ldots,\omega_n$ such that
\[
   e=\sum_{i=1}^n \omega_i,
\]
and if an analogous decomposition exists for $u_A-e$. This definition does not assume that the $\omega_i$ are mutually orthogonal; however, they
have to be as a consequence of Postulates 1 and 2. To see this, note that for all $j$
\[
   1=\langle u_A,\omega_j\rangle\geq\langle e,\omega_j\rangle=1+\sum_{i\neq j}\underbrace{\langle \omega_i,\omega_j\rangle}_{\geq 0},
\]
hence $\langle\omega_i,\omega_j\rangle=0$ for all $i\neq j$. The corresponding effects $e$ can also be characterized in two further ways,
namely as \emph{projective units} and as the extremal points of the unit order interval, giving further weight to the interpretation as the analogue
of orthogonal projectors in quantum theory. This is the content of the next lemma. We start with a definition.
\begin{definition}[Projective units]
\label{DefProjectiveUnits}
Let $A$ be any system satisfying Postulates 1 and 2. Then, for every face $F$ of $A_+$, define the \emph{projective unit} $u_F$ as
\[
   u_F:=P_F u_A,
\]
where $P_F$ is the orthogonal projection onto the linear span of $F$. A projective unit $u_F$ is called \emph{atomic} if $|F|=1$.
\end{definition}

This is now used in the following lemma:

\begin{lemma}
\label{LemProjectiveUnits}
Let $A$ be any system satisfying Postulates 1 and 2. Then, for every face $F$ of $A_+$, there is a unique effect $u_F$ with $0\leq u_F\leq u_A$
such that $u_F(\omega)=1$ for every $\omega\in F\cap \Omega_A$, and $u_F(\varphi)=0$ for all $\varphi\in F'\cap \Omega_A$, namely
the projective unit from Definition~\ref{DefProjectiveUnits}.
If $\omega_1,\ldots,\omega_{|F|}$
is any frame that generates $F$, then
\begin{equation}
   u_F=\sum_{i=1}^{|F|} \omega_i.
   \label{eqFrameDecomp}
\end{equation}
Furthermore, every effect $e\in A_+$ with $0\leq e \leq u_A$ is a convex combination of projective units, and we have $u_F+u_G\leq u_A$
if and only if $F\perp G$, in which case $u_F+u_G=u_{F\vee G}$.
\end{lemma}

\begin{proof}
As in Definition~\ref{DefProjectiveUnits}, set $u_F:=P_F u_A$. Due to Theorem~\ref{ThmPosProjection}, $u_F\in A_+$. Thus, $\omega\in F$ implies
\[
   \langle u_F,\omega\rangle = \langle P_F u_A,\omega\rangle=\langle u_A,P_F\omega\rangle=\langle u_A,\omega\rangle=1.
\]
If $\varphi\in F'$, then $P_F\varphi=0$, and an analogous computation shows that $\langle u_F,\varphi\rangle=0$.
Set $\mu_F:=u_F/\langle u_A,u_F\rangle$, then $\mu_F\in F\cap\Omega_A$, and so there is a frame $\omega_1,\ldots,\omega_{|F|}$ of $F$
such that $\mu_F=\sum_{i=1}^{|F|}\lambda_i\omega_i$ with $\lambda_i\geq 0$, $\sum_i\lambda_i=1$. For every $j=1,\ldots,|F|$, we have $\omega_j\in F$, and so
\[
   1=\langle u_F,\omega_j\rangle=\langle u_A,u_F\rangle \sum_{i=1}^{|F|} \lambda_i \langle \omega_i,\omega_j\rangle=\lambda_j\langle u_A,u_F\rangle,
\]
so all $\lambda_j$ are equal to $\langle u_A,u_F\rangle^{-1}$, proving that there exists \emph{some} frame $\omega_1,\ldots,\omega_{|F|}$
with decomposition~(\ref{eqFrameDecomp}) of $u_F$, and  showing the inequality $0\leq u_F\leq u_A$. If $\varphi_1,\ldots,\varphi_{|F|}$ is any other frame on $F$,
then there exists a reversible transformation $T$ with $T\omega_i=\varphi_i$. Since both frames generate $F$, $T$ must preserve the face $F$ (and also
its orthogonal complement because $T$ is orthogonal). Hence
\[
   u_F+u_{F'}=u_A=T u_A = T u_F+ T u_{F'}.
\]
Thus $u_F=T u_F=T\sum_{i=1}^{|F|}\omega_i = \sum_{i=1}^{|F|}\varphi_i$, proving that $u_F$ can be decomposed into \emph{any} frame in the claimed way.
If $0\leq e\leq u_A$ is any effect, then it has a frame decomposition $e=\sum_{i=1}^{|A_+|}\lambda_i \omega_i$, where $\omega_i\in\Omega_A$ are
mutually orthogonal pure states, and $0\leq \lambda_i\leq 1$. Thus, the vector $\lambda:=(\lambda_1,\ldots,\lambda_{|A_+|})$ is an element of the $|A_+|$-dimensional
unit cube, and can thus be written as a convex combination of extremal points of the (convex) cube, corresponding to vectors $\mu=(\mu_1,\ldots,\mu_{|A_+|})$ where all
$\mu_i\in\{0,1\}$. Hence $e$ can correspondingly be decomposed into effects of the form $\sum_{i=1}^{|A_+|} \mu_i \omega_i$, which are projective units.
This also shows that the $u_F$ are the unique effects with the properties stated in the lemma. If $F\perp G$ then $u_F+u_G=u_{F\vee G}\leq u_A$ is clear from the sum
representation of projective units; conversely, if $u_F+u_G\leq u_A$, then $u_F+u_{F'}=u_A$ implies that $u_F+u_G\leq u_F+u_{F'}$, and
so $u_G\leq u_{F'}$. Thus, if $\omega\in G\cap \Omega_A$, then $1=\langle u_G,\omega\rangle\leq\langle u_{F'},\omega\rangle\leq 1$,
and so $0=\langle u_A-u_{F'},\omega\rangle=\langle u_F,\omega\rangle=\langle P_F u_A,\omega\rangle
=\langle u_A,P_F \omega\rangle$, which implies that $P_F\omega=0$ and $\omega\perp F$. Hence $F\perp G$.
\end{proof}

Following the definition of~\cite{CabelloSeveriniWinter}, expressed in the language of~\cite{Fritz}, every system satisfying Postulates 1 and 2
defines a \emph{contextuality scenario} given by a hypergraph $H$, where the vertices of $H$ are the projective units $u_F$ ($F\neq\{0\}$ any face of $A_+$),
and the edges are collections of effects $u_{F_1},\ldots, u_{F_n}$ with $\sum_{i=1}^n u_{F_i}=u_A$. These edges describe contexts, i.e.\ sharp
measurements (given by sets of projective units) that are compatible (i.e.\ jointly measurable).
\begin{theorem}
Any system satisfying Postulates 1 and 2 also satisfies the principle of \emph{Consistent Exclusivity} $\mbox{CE}^1$
as given in~\cite[Def.\ 7.1.1]{Fritz} and~\cite{Henson}.
\end{theorem}
\begin{proof}
We have to show the following: if $I$ is any set of vertices of the hypergraph $H$ such that every two elements of $I$ belong to a common edge, then $\sum_{e\in I} e(\omega)\leq 1$
for all $\omega\in \Omega$. In the context of Postulates 1 and 2, $I$ is then a set of projective units
$I=\{u_{F_1},\ldots,u_{F_n}\}$ such that $u_{F_i}+u_{F_j}\leq u$ for $i\neq j$.
But Lemma~\ref{LemProjectiveUnits} implies that $F_i\perp F_j$. So if $\mathcal{F}_i$ is any frame for $F_i$, then
$\mathcal{F}_i\perp\mathcal{F}_j$ for $i\neq j$, hence the disjoint union $\mathcal{F}:=\bigcup_i \mathcal{F}_i$ is a frame on $A_+$, generating some face $F$. Thus
\[
   \sum_{e\in I}e(\omega)=\sum_{i=1}^n \langle u_{F_i},\omega\rangle=\sum_{i=1}^n\sum_{e\in\mathcal{F}_i} \langle e,\omega\rangle=\langle u_F,\omega\rangle\leq 1.
\]
This proves the claim.
\end{proof}

As mentioned in Section~\ref{SecIntroduction}, the classification of the
set of all state spaces that satisfy Postulates 1 and 2 remains an
open problem with interesting physical and mathematical
implications.  Now we show that one additional assumption brings us
into the realm of Jordan algebra state spaces. Before postulating the
absence of third-order interference, we study another postulate which
turns out to be equivalent in our context.

\section{Jordan systems from Postulates 1+2 and purity preservation by filters}
\label{SecJordan}

In this section, we show that a system satisfying
  Postulates 1 and 2 and a third postulate, that the positive
  projections of Theorem \ref{ThmPosProjection} take pure states to
  multiples of pure states, is either an irreducible Jordan algebraic
  system or classical.

Jordan algebras were introduced around 1932 by Pascual 
Jordan \cite{Jordan33}, as a
  potentially useful algebraic abstraction of the space of
  \emph{observables}, i.e. Hermitian operators on a Hilbert space, in
  the newly minted quantum theory.  Since the usual matrix or operator
  multiplication does not preserve Hermiticity, its physical
  significance was unclear; Jordan focused on abstracting properties
  of the symmetrized product $A \bullet B := (AB + BA)/2$ which
  does preserve Hermiticity.  Like the space of Hermitian operators, a
  Jordan algebra (as initially defined by Jordan and studied by him,
  von Neumann, and Wigner) is a real vector space, closed under a
  commutative bilinear product $\bullet$.  Since the symmetrized
  product of Hermitian operators is not associative but does satisfy
  the special case $(a^2 \bullet b) \bullet a = a^2 \bullet (b \bullet
  a)$ (where $a^2 := a \bullet a$) of associativity, a Jordan algebra
  is not assumed associative, but only to satisfy this special case,
  the ``Jordan property''.  
For a finite-dimensional Jordan algebra $A$, at least, the
  squares (elements of the form $a^2$ for some $a \in A$) form a
  closed cone of full dimension. Jordan, von Neumann, and Wigner
  investigated the \emph{formally real} finite-dimensional Jordan
  algebras, which are precisely those whose cones of squares are
  pointed.  
Like the quantum observables, formally real Jordan
  algebras have a well-behaved spectral
  theory (see~\cite[Sec.\ III.1]{FarautKoranyi}), with real spectra 
and an associated real-valued
 trace function.\footnote{In finite dimensions, formal reality coincides with the notion of Euclideanity, used in the references
\cite{FarautKoranyi} and \cite{ASBook}.}  
In these algebras, squares have nonnegative spectra, and
  the unit-trace squares form a closed compact convex set as required
  to be the normalized state space of a system in our context.  
As mentioned in the
  introduction, the finite-dimensional formally real Jordan algebras
  are already quite close to quantum theory: besides standard quantum
  theory over the complex numbers they are quantum-like systems over
  the reals and over the quaternions, systems whose state spaces are
  balls (``spin factors'') and what can be thought of as
  three-dimensional quantum theory over the octonions \cite{JvNW1934}.
  They are also of interest because they are precisely the
  finite-dimensional systems whose cones of unnormalized states are
  self-dual and homogeneous \cite{Koecher1958, Vinberg1960}.

The key tools we will use to establish the main result of this
section are Theorem \ref{ThmPosProjection} 
and a characterization of the state spaces of certain Jordan algebras by
Alfsen and Shultz~\cite[Thm.\ 9.33]{ASBook}, first published in~\cite{ASpaper}.
To state this result requires introducing several somewhat technical
notions, which are, however, of considerable physical interest in
their own right.  These are the notions of a \emph{filter} on the
state space $A$ (and its dual, the notion of a \emph{compression}
on the effect space $A^*$), with its associated notion of a projective
state space, and the property of \emph{symmetry of transition
  probabilities}.

We first define filters, and begin by introducing some notions used in 
that definition.      

\begin{definition}\label{def: complementary projections}
Let $A$ be any state space with cone $A_+$.  \emph{Projections} are
linear operators $P:A\to A$ with $P^2=P$; they are \emph{positive} if $P(A_+)\subseteq A_+$. Positive projections $P$
and $Q$ are called \emph{complementary} if $\im_+ P = \ker_+ Q$ and
vice versa, where $\im_+ P:=\im P \cap A_+$ and $\ker_+ Q:=\ker Q \cap
A_+$. A positive projection $P$ is \emph{complemented} if there exists
a positive projection $Q$ such that $P$ and $Q$ are complementary.
\end{definition}

\begin{definition}[Filters and projectivity] \label{def: filters and projectivity} 
A \emph{filter} is a positive linear  projection $P:A\to A$ which (i)
is complemented, 
(ii) has a complemented dual $P^*$, and
(iii) is normalized, i.e.
satisfies $u_A(P\omega)\leq u_A(\omega)$ for all $\omega\in A_+$.
\footnote{This condition is equivalent to base norm contractiveness,
  which is what Alfsen and Shultz use in their definition.  In the
  Appendix to \cite{ASBook}, item A24, they define, for $\omega,
  \sigma \in V_+$, $V$ a base norm space, $\omega \perp \sigma$ by
  $||\omega - \sigma|| = ||\omega|| + ||\sigma||$.  A26 states that
  each $\rho \in V$ can be decomposed as a difference of two
  orthogonal positive components, i.e. there are $\omega, \sigma \in
  V_+$ such that $\omega \perp \sigma$ and $\rho = \omega - \sigma$.
  From this we can see that base-norm contractiveness ($||T \rho ||
  \le ||\rho||$) of a map $T$ on $V_+$ implies contractiveness
  everywhere. Since $\|\omega\|=u_A(\omega)$ for all $\omega\in V_+$, we
  have equivalence of base norm contractiveness and normalization of filters.}

The state space $A$ is called \emph{projective} if every face of $A_+$
is the positive part, $\im_+ P$, of the image of a filter $P$.
\end{definition}

We define filters in order to make use of the results in
\cite{ASBook}, but they are also of great interest in their own right.
Actually Alfsen and Shultz define
\cite[Def.\ 7.22]{ASBook} \emph{compressions}, acting on the effect
space $A^*$.  The finite-dimensional specialization of Alfsen and
Shultz' notion of compression is just a positive projection $Q: A^*
\rightarrow A^*$ which is complemented, whose dual is complemented,
and whose dual is normalized; it is obvious that a linear map $Q: A^*
\rightarrow A^*$ is a compression iff $Q^*: A \rightarrow A$ is a
filter, and similarly $P$ is a filter iff $P^*$ is a compression.  
We defined filters because we are most interested in the
transformations that act on the state space $A$.  In fact, in the
context of Postulates 1 and 2 with $A$ and $A^*$ being identified via
an appropriate self-dualizing inner product, filters and compressions
are represented by precisely the same linear operators.

As described above, in standard quantum theory the face associated
with a subspace $S$ of Hilbert space consists of the density matrices
whose support is contained in $S$.  Quantum state spaces are
projective: there is a filter onto each face, namely the linear map
$\rho\mapsto P_S \rho P_S$, where $P_S$ is the orthogonal projector
onto $S$.  The complementary projection is $\rho\mapsto P_{S^\perp}\rho P_{S^\perp}$.

One of several reasons that filters are of great interest for 
physics and information-processing is that they
share with the maps $\rho \mapsto P \rho P$ the property of \emph{neutrality}
\cite[Def. 7.19]{ASBook}:
if a state $\omega \in \Omega$ ``passes the filter with probability $1$'', 
i.e. $u_A(P \omega) = u_A(\omega)$, then it ``passes the filter undisturbed'', 
i.e. $P \omega = \omega$.
(This is immediate from Definition 7.19, Proposition 7.21, 
and Definition 7.22 of \cite{ASBook}.) 

We now turn to \emph{symmetry of transition probabilities}, a notion
which is defined for systems which are projective in the sense of
Definition \ref{def: filters and projectivity}.

Observe that in a projective system, for each atomic projective unit
$p$, which is associated~\cite[Prop.\ 7.28]{ASBook} with a unique
filter $P$ for which $P^* u = p$, the associated face $\{x\,\,|\,\,
p(x)=1\}$ of $\Omega$ contains a single pure state.  Call this state
$\hat{p}$.  The map $p \mapsto \hat{p}$ is a one-to-one map from the
set of atoms of the lattice of projective units onto the set of
extremal points of $\Omega$.  The system is said to satisfy
\emph{symmetry of transition probabilities} \cite[Def.\ 9.2
  (iii)]{ASBook} if for all pairs $a,b$ of atoms of the lattice of
projective units, $a(\hat{b}) = b(\hat{a})$.

\begin{lemma}\label{lemma: symmetry of transition probabilities}
If a system satisfies Postulates 1 and 2, it satisfies 
symmetry of transition probabilities.
\end{lemma}
\begin{proof}
In the context of Postulates 1 and 2,
atomic projective units are $u_F$ for $|F|=1$, where $F$ is generated by a pure state (frame of size $1$) $\omega_1$,
such that $u_F(\varphi)=\langle\omega_1,\varphi\rangle$ according to Lemma~\ref{LemProjectiveUnits}, so $\hat u_F=\omega_1=u_F$ in the
notation just introduced. Thus $a(\hat b)=\langle \hat a,\hat b\rangle = \langle \hat b,\hat a\rangle=b(\hat a)$.
\end{proof}

We can now state a version of a theorem from \cite{ASBook} that we will use
in proving the main result of this section.  One of the conditions in
this theorem will be important in its own right in what follows, and
we therefore call it \emph{Postulate 3}${}^\prime$.

\begin{theorem}
\label{ThmAlfsenShultz}

Let a finite-dimensional system $A_+$ satisfy
\begin{itemize}
\item[(a)] projectivity,
\item[(b)] symmetry of transition probabilities, and
\item[(c)] \textbf{Postulate 3${}^\prime$:} filters $P$ preserve purity. That is, if $\omega$ is a pure state, then $P\omega$ is a nonnegative multiple of a pure state.
\end{itemize}
Then $A_+$ is the state space of a formally real Jordan algebra.
\end{theorem}
The original theorem in~\cite{ASBook,ASpaper} is formulated in terms
of compressions, with similar results in finite dimensions given by
Gunson~\cite{GunsonAlg67} and by
Guz~\cite{GuzFilter78,GuzPure79,GuzCondProb80}.
Theorem~\ref{ThmAlfsenShultz} above is an adaptation to our language
and to finite dimension, using the notion of filters instead of
compressions.  The conjunction of (b) and (c) is what Alfsen and
Shultz~\cite[Def. 9.2]{ASBook} call the ``pure state properties''
(their (3)), while their (2) is a technical condition that is
automatically satisfied in finite dimension, and their (1) follows
from our (a).

\begin{theorem} \label{first main theorem}
In finite dimension, Postulates 1 and 2 imply that the system is projective.
Assuming in addition Postulate 3$^{~\prime}$ implies that the system is either irreducible Jordan-algebraic, or classical.
\end{theorem}
\begin{proof}

In Theorem~\ref{ThmPosProjection},
we have already shown that the orthogonal projection $P_F$ 
onto the linear span of the face $F$ is
positive, for every face $F$. Now we show that it is a  
filter, which establishes that
$A_+$ is projective.  For any face $F$, the corresponding projection
$P_F$ satisfies $\im P_F=\lin F$, $\im_+ P_F=F$, and $\ker_+ P_F=F' =
F^\perp \intersect A_+$.  So also $\im_+ P_{F'} = F'$ and $\ker_+
P_{F'} = F'' = F$, and we see that $P_F$ and $P_{F'}$ are complements,
establishing property (i) in the definition of filter.  Since 
$P_F = P_F^*$, $P_F$ has complemented adjoint, property
(ii).  $P_F$ and $P_{F'}$ are positive by Theorem
\ref{ThmPosProjection}.  To see property (iii), i.e.\ normalization of
$P_F$, recall from Lemma~\ref{LemProjectiveUnits} that $u_A(P_F
\omega) = u_F(\omega)\leq u_A(\omega)$.  Hence for every face $F$ the 
projection $P_F$ is a filter, so the system is projective.

Projectivity is (a) of Theorem~\ref{ThmAlfsenShultz}.  
Lemma \ref{lemma: symmetry of transition probabilities} states that 
condition (b) of Theorem~\ref{ThmAlfsenShultz} follows from Postulates 1 and 2. 
So (a) and (b) of that theorem follow from Postulates 1 and 2, whence
by the theorem, Postulates 1, 2, and purity preservation by filters imply 
that a system is Jordan algebraic.

To see that the only reducible Jordan-algebraic cones this allows are
the classical ones (corresponding to direct sums of the
one-dimensional formally real Jordan algebra), note that the cone of a
direct sum of Jordan algebras $\mathcal{A}:=\bigoplus_{i=1}^n \mathcal{A}_i$
is the direct sum $C=\bigoplus_{i=1}^n C_i$ of their cones. This is because every $a\in\mathcal{A}$
can then be written $a=(a_1,\ldots,a_n)$, and the elements of $C$ are the squares
$a^2=(a_1^2,\ldots,a_n^2)$, where the single $a_i^2$ entries range over all of $C_i$.
Suppose one of the summands, say $C_j$, is not one-dimensional.  The
face generated by two ray-extremal points, $\omega_j \in C_j$ and
$\omega_k \in C_k$, with $k \ne j$, is a direct sum of one-dimensional
cones, i.e. a classical bit.  Since $C_j$ is irreducible and not
one-dimensional, it is not classical, so it contains perfectly
distinguishable pure states $\omega_j$ and $\omega_i$ that generate a
face that is not a direct sum.  Since we have another rank-2 face that
is a direct sum, in light of Proposition \ref{FramesAndFaces} 
this violates Postulate 2.  Hence either the cone is
irreducible, or all summands are one-dimensional (i.e. it is
classical).
\end{proof}

The following proposition will be needed later.
\begin{proposition}\label{propPeirceDecomposition}
Assume Postulates 1 and 2. Then, to every face $F_1$ of $A_+$ with complementary face $F_2\equiv F'_1$ and corresponding projections
$P_1$ and $P_2$, the space $A$ has an orthogonal decomposition
\[
   A = A_1 \oplus A_2 \oplus A_{12}^c,
\]
where $A_i := \im P_i$, $A_{12}^c := \ker P_1 \intersect \ker P_2$.
\end{proposition}
\begin{proof}
By construction, $A_1=\lin F_1\perp \lin F'_1=A_2$, and by elementary linear algebra,
$(A_1\oplus A_2)^\perp = A_1^\perp \cap A_2^\perp = \ker P_1 \cap \ker P_2 = A_{12}^c$.
\end{proof}

\section{Third-order interference}
\label{SecInterference}

Rafael Sorkin defined a notion of $k$-th order interference
\cite{SorkinMeasure94}, which can be manifested in analogues of the
two-slit experiment involving $k$ or more slits.  This notion was
adapted to projective convex systems in \cite{InterferenceInPrep, CozThesis},
and the $k=3$ case explored in \cite{InterferenceShort}.
Quantum theory exhibits $k=2$ interference, but
no higher interference.  In this section, we show that Postulates 1
and 2, plus the assumption of no third-order interference,
characterize irreducible Jordan algebraic systems.

We formalize
  the assumption of no third-order interference using a mathematical
  definition of $M$-slit interference experiment given in terms of
  experimental probabilities.  This is motivated by, and abstracted
  from, specific concrete experimental interference experiments such
  as those in which a photon passes through physical slits in a
  barrier, but the probabilistic definition gives a conceptual account
  of the notion of interference that applies (as does the usual
  quantum-mechanical concept of interference) far more broadly.
Consider the setup depicted in
  Figure~\ref{fig_interference}; we would like to give a formal
  description of the experimental behavior, given that a certain
  subset of the slits is open or blocked. First, imagine the case that
  all slits are open, and consider the state $\omega$ of the particle
  immediately after it has passed the slit arrangement. By preparing
  the particle in different ways, we can obtain different states
  $\omega$. Part of the state $\omega$ contains the ``which-slit
  information'', encoding through which slit the particle has just
  passed the arrangement (possibly in a probabilistic mixture or
  generalized superposition). In an ideal
  $M$-slit experiment we would in principle
  be able to measure through which slit the
  particle passes, if we put suitable
  detectors behind the slits.

For every slit $j\in\{1,2,\ldots,M\}$, there  should 
exist states $\omega$ such that the particle is definitely found at slit $j$, if measured.
In our mathematical setting, this means that there is a face $F_j$ of the state space, such that all states $\omega\in F_j$ give unit probability
for the ``yes''-outcome of the two-outcome measurement \emph{``is the particle at slit $j$''?} Moreover, the slits should be perfectly distinguishable -- if
a particle is definitely at slit $j$, then it is definitely not at slit $i$ for all $i\neq j$. Mathematically, this means that $F_i\perp F_j$ for $i\neq j$.

We can also ask coarse-grained questions like \emph{``Is the particle found among slits $1$ and $2$ (rather than somewhere else)?} The set
of those states $\omega$ that give unit probability for the ``yes''-outcome must contain both $F_1$ and $F_2$; therefore, it must contain
$F_1\vee F_2$, the smallest face of the state space that contains both $F_1$ and $F_2$ as subsets. Furthermore, it should
be the smallest such face, since we
do not want to include further possibilities. Thus, this set will be $F_{12}:=F_1\vee F_2$. More generally, for every subset of slits $J\subseteq\{1,2,\ldots,M\}$,
we have a face $F_J=\bigvee_{j\in J}F_j$, containing those states that describe a particle that will definitely be found 
to be somewhere among the slits in $J$ if the corresponding effect is measured.
If the setup is ``complete'' in the sense that every particle must definitely be found at one of the slits if measured, the face $F_{12\ldots M}$
must be the full state space.

Now imagine an additional detector following the slit arrangement, as depicted in Figure~\ref{fig_interference}. It may click or not click;
the probability to click in the case that all slits are open is described by some effect $e$. Suppose we block all slits, except for
a subset $J\subset\{1,\ldots,M\}$ of the slits which are left open. The combination of blockings and detector defines a new measurement,
given by some other effect $e_J$, with click probability $e_J(\omega)$ if the state right before the blockings is $\omega$.

If the slits do what we intuitively expect them to do,
  as they do to a good approximation in quantum-mechanical
  multi-slit experiments, then the click probabilities should behave
as follows. If $\omega$ is a state of a particle that
would definitely be found at one of the
slits among $J$, i.e.\ $\omega\in F_J$, then the blockings should have
no effect (because the slits $J$ are all open), and the click
probability should remain the same: $e_J(\omega)=e(\omega)$. On the
other hand, if the particle would definitely not be found among the
open slits $J$, i.e.\ $\omega\in F'_J$, then the
particle should be blocked and there should definitely be no detector
click, and $e_J(\omega)=0$.

These considerations lead to the following definition, which abstracts
probabilistic properties of an interference experiment from 
particular physical realizations involving slits, spatial paths, and so forth.
We will soon see that the orthogonal projections $P_J$ onto the faces $F_J$
are of paramount importance, which is why we introduce a name for them as well.

\begin{definition}[$M$-slit experiment]
\label{DefMSlitExperiment}
A set of effects $e_J$ and faces $F_J$, $J\subseteq \{1,2,\dots,M\}$, with $F_J=\bigvee_{j\in J}F_j$ and $F_i\perp F_j$ for $i\neq j$
is called an \emph{$M$-slit experiment} if there is an effect $e\in[0,u]$ such that
\begin{itemize}
\item $e_J(\omega)=e(\omega)$ for all $\omega\in F_J$,
\item $e_J(\varphi)=0$ for all $\varphi\in F'_J$.
\end{itemize}
For any given set of faces with the properties stated above, the corresponding set of orthogonal
projections $P_J:=P_{F_J}$ will be called an \emph{$M$-slit mask}. It is called \emph{complete} if $F_{12\cdots M}=\Omega_A$,
that is, if $P_{12\cdots M}=\mathbf{1}$.
\end{definition}

Such an experiment exhibits second-order interference (say, for $M=2$)
if the overall interference pattern $e_{12}(\omega)$ fails to be the
sum of the one-slit patterns $e_1(\omega)$,
$e_2(\omega)$.
If it exhibits second-order interference, it may in addition exhibit
irreducibly third-order interference. Third-order interference occurs if the
overall pattern $e_{123}(\omega)$
fails to be the sum of the double-slit patterns $e_{ij}(\omega)$, corrected for
overcounting by subtracting suitable multiples of the single-slit patterns $e_i(\omega)$.
Unless otherwise specified we
use the notation $\sum_{i < j}$ to mean the double sum $\sum_{i} \sum_{j>i}$.

\begin{definition}[Third-order interference]
\label{DefThirdOrder}
We say that a state space exhibits \emph{third-order interference} if
there exists an $M$-slit experiment (for some $M\geq 3$)
and a state $\omega$ such that
\begin{equation}
   e_{12\ldots M}(\omega)\neq \sum_{i<j} e_{ij}(\omega)-(M-2)\sum_i e_i(\omega).
   \label{eqDef3rdOrderEffects}
\end{equation}
In particular, for $M=3$,  the condition is
\[
   e_{123}(\omega)\neq e_{12}(\omega)+e_{13}(\omega)+e_{23}(\omega)-e_1(\omega)-e_2(\omega)-e_3(\omega).
\]
\end{definition}

The second term in~(\ref{eqDef3rdOrderEffects}) corrects for the
overlaps of the sets $\{i,j\}$ as each index occurs
$M-1$ times in pairs $i < j$.
Sorkin's \cite{SorkinMeasure94} original definition, and the
discussion in \cite{InterferenceInPrep, InterferenceShort}, used the
$M=3$ case as their definition of third-order interference, but the
two can straightforwardly if somewhat tediously be shown to be
equivalent.  Sorkin showed that if a scenario lacks $k$-th order 
interference, it cannot have $l$-th order interference for any $l > k$.

With the previous definition, we can give a concise formal statement of Postulate 3:
\begin{postulate}State spaces do not exhibit third-order interference, as introduced in Definition~\ref{DefThirdOrder}.
\end{postulate}
Now we show that $M$-slit experiments are closely related to the positive orthogonal projections introduced in Theorem~\ref{ThmPosProjection}.

\begin{proposition}
\label{PropInterferenceProjectors}
Assume Postulates 1 and 2. Then, given any $M$-slit experiment with
effects $e, e_J$, we have $e_J=P_J e$, where the $P_J$ are the
elements of the corresponding $M$-slit mask.
Conversely, given any set of faces
$F_J$, $J\subseteq\{1,2,\ldots,M\}$, with
$F_J=\bigvee_{j\in J}F_j$ and $F_i\perp F_j$ for $i\neq j$
and any effect $e$, the set of effects
$e_J:=P_J e$ defines an $M$-slit experiment.
\end{proposition}
\begin{proof}
Since $\langle e-e_J,\omega\rangle=0$ for all $\omega\in F_J$, we have
$e-e_J\in F_J^\perp$ (not necessarily in $F'_J$, because we do not
yet know whether $e-e_J$ is positive). Similarly, $\langle
e_J,\varphi\rangle=0$ for all $\varphi\in F'_J$ and $e_J\geq 0$
implies that $e_J\in F_J$. Thus
\[
   P_J e = P_J ( e_J+(e-e_J)) = e_J.
\]
The converse can be checked by direct calculation.
\end{proof}

According to this proposition, absence of third-order interference can
be expressed in terms of the orthogonal projections only:
\begin{lemma}
\label{LemEquivThirdOrder}
Consider a state space satisfying Postulates 1 and 2. It has no
third-order interference if and only if for any $M$-slit mask $P_J$, $J\subseteq\{1,\ldots,M\}$, it holds that
\beq \label{eqDef3rdOrder} P_{12\cdots M} = \sum_{i < j}
P_{ij} - (M-2) \sum_i P_i.  \eeq
\end{lemma}
\begin{proof}
We have absence of third-order interference if for any choice of faces (as described in the statement of the lemma)
and choice of effect $e$ as well as state $\omega$, (\ref{eqDef3rdOrderEffects}) holds with equality. Since the states span the
space, this is equivalent to the statement
\[
   e_{12\ldots M}= \sum_{i<j} e_{ij}-(M-2)\sum_i e_i,
\]
and, due to Proposition~\ref{PropInterferenceProjectors}, to
\[
   P_{12\ldots M} e = \sum_{i<j} P_{ij} e -(M-2)\sum_i P_i e.
\]
As this must hold for all effects $e$, and the effects span the space, we obtain the statement of the lemma.
\end{proof}

Now we are ready to prove one of 
our main results about the absence of third-order interference together with Postulates 1 and 2:

\begin{theorem}\label{TheoremInterference}
A system satisfies 
Postulates 1, 2 and 3 if and only if
it is an irreducible Jordan system or a classical system.
\end{theorem}

\begin{proof} 
We begin with the ``if'' direction: irreducible Jordan
systems and classical systems satisfy Postulates 1, 2 and 3.  For
classical systems it is well-known and easy to see that Postulates 1
and 2 are satisfied: indeed, finite-dimensional classical state spaces
$\Omega$ are often \emph{defined} as those for which every state has a
\emph{unique} decomposition into extremal points, and in this case
Postulate 2 follows from the fact that any permutation of the extreme
points in this unique maximal frame is an affine automorphism of
$\Omega$.  Classical systems do not even have 2nd-order interference
\cite{SorkinMeasure94} (the first level that is actually
interference), so they cannot have any higher order of
  interference.  It follows directly from a fairly standard
orthogonal decomposition in formally real Jordan algebras
(see e.g.~\cite{FarautKoranyi}) that finite-dimensional Jordan systems
satisfy Postulate 1; and it is also well-known that the Jordan algebra
automorphisms are affine automorphisms of the normalized state space,
and act transitively on the set of ordered sets of orthogonal extremal
states in the irreducible case~\cite{FarautKoranyi}.
In Proposition~\ref{interference and purity preservation} below, we show that in the context of Postulates 1 and 2, absence of third-order
interference is equivalent to the property that filters preserve purity of states. Since the latter property is well-known for a class of Jordan
systems including the finite-dimensional ones~\cite[Thm.\ 9.38]{ASBook}, this shows that they also satisfy Postulate 3.

The ``only if'' direction is an immediate 
consequence of Proposition~\ref{interference and purity preservation}---to
be proved in the remainder of this section---which 
states that the absence of third-order interference implies that all
filters preserve purity, together with Theorem \ref{first main theorem}, 
which states that Postulates 1, 2, and purity-preservation by filters 
imply that systems are irreducible Jordan, or classical.
\end{proof}

We could have defined an $M$-slit experiment directly in terms of the positive projections $P_J$
onto the faces. These describe the action of the slits on the state. However, referring to the corresponding effects $e_J$
in Definition~\ref{DefMSlitExperiment} has the advantage that we know for sure that the effects can be implemented
(due to Proposition~\ref{SpectralityImpliesSaturation}). On the other hand, there is no analogous statement that guarantees
that the projections $P_J$ themselves can actually be implemented as physical transformations. Thus, not referring to
positive projections in the definition of an $M$-slit experiment means that we make fewer assumptions.

The proof of the crucial Proposition~\ref{interference and purity preservation}
proceeds via several other propositions and lemmas. The following property is also
mentioned in~\cite{CozThesis} and~\cite{InterferenceShort}.

\begin{lemma}\label{LemmaCompatibleFiltersCommute}
It follows from Postulates 1 and 2 that $P_J P_K = P_{J\cap K}$ for any $M$-slit mask.
\end{lemma}
\begin{proof}
First note that if $F$, $G$ and $H$ are faces such that $F\perp H$ and $G\perp H$, then $(F\vee G)\perp H$.
This is because $(F\vee G)\cap H^\perp$ is a face which contains $F$ and $G$, and is also a subset of $F\vee G$,
hence equal to $F\vee G$.

Defining the projective units $u_j:=u_{F_j}$ and $u_J := u_{F_J}$, it follows from Lemma~\ref{LemProjectiveUnits}
that $u_K=\sum_{k\in K} u_k$. Hence
\[
   u_K = \sum_{k\in K\cap J} u_k + \sum_{l\in K\setminus J} u_l.
\]
For $j\in J$ and $l\in K\setminus J$ we have $F_j\perp F_l$, thus $F_J=\vee_{j\in J} F_j \perp F_l$, and so $P_J u_l=0$.
On the other hand, if $k\in K\cap J$ then $P_J u_k=u_k$, so
\[
   P_J u_K = \sum_{k\in K\cap J} u_k \leq \sum_{k\in K} u_k = u_K.
\]
According to~\cite[Prop.\ 7.39]{ASBook}, this implies that $P_J P_K= P_K P_J$, which in turn implies~\cite[Thm.\ 8.3]{ASBook}
that $P_J P_K = P_J\wedge P_K = P_{J\cap K}$.
\end{proof}

The next proposition uses the decomposition
described in Proposition~\ref{propPeirceDecomposition} to derive a
similar decomposition corresponding to a complete $M$-slit mask.

\begin{proposition}\label{prop: M slit Peirce decomposition}
Let $P_i$ with $i \in \{ 1,..., m\}$ be a complete $M$-slit mask on a system $A$ satisfying Postulates 1 and 2.  Then there
is an orthogonal decomposition
\beq \label{M slit Peirce decomposition}
A = \oplus_i A_i \oplus_{i<j} A^c_{ij} \oplus A^{(3)}
\eeq
where $A_i := \im P_i, A^c_{ij} := \ker P_i \intersect \ker P_j \intersect
\im P_{ij}$ and $A^{(3)} := \bigcap_{i < j}  \ker P_{ij}$.
\end{proposition}

\begin{proof}
Using Proposition~\ref{propPeirceDecomposition} and the fact that
each face is itself a system satisfying Postulates 1 and 2, we decompose each
$\im P_{ij}$ as $A_i \oplus A_j \oplus A_{ij}^c$.  (Note that we still
have $A_{ij}^c$ orthogonal to $A_i \oplus A_j$ because it is contained
in $\ker P_i \intersect \ker P_j$.) 
For $k, l \notin \{i,
j\}$ we have $A_{ij}^c \perp A_{kl}^c$, since $\im P_{ij} \perp \im
P_{kl}$.  Furthermore, for $i \neq k$, $A^c_{ij} \perp A^c_{jk}$,
because for $x \in A^{c}_{ij}, y \in A^{c}_{jk}$
\beqa
\langle x , y \rangle &=& \langle P_{ij} x, P_{jk} y \rangle
\nonumber \\ &=& \langle x,  P_{ij} P_{jk} y \rangle
= \langle x,  P_j y \rangle = 0\;, \nonumber
\eeqa
where the first equality follows from $x \in \im P_{ij}, y \in \im P_{jk}$
due to the definitions of
$A^{c}_{ij}, A^{c}_{jk}$, the last equality from $y \in \ker P_j$ due
to the definition of $A^{c}_{jk}$, and the second last equality from Lemma
\ref{LemmaCompatibleFiltersCommute}.
Now we just have to show that $A^{(3)} := \bigcap_{i < j} \ker P_{ij}$ is
the orthogonal complement of $\oplus_i A_i \oplus_{i < j} A^c_{ij}$.  Since
$\oplus_i A_i \oplus_{i < j} A^c_{ij} = \lin \{ \bigcup_{i < j} \im P_{ij} \}$,
$(\oplus_i A_i \oplus_{i < j} A^c_{ij})^\perp =
\bigcap_{i < j} \ker P_{ij}$, and we are done.
\end{proof}

It is interesting to note that the pairwise intersections $\ker P_i \intersect \ker P'_i$ represent ``coherences'' associated with
the two-slit experiment $P_i, P'_i$ \cite{InterferenceShort}, and that intersecting this with $\im
P_{ij}$ gives the part associated with the two-slit
experiment $P_i, P_j$. As an example, consider a quantum $3$-level system with orthonormal basis
$\{|i\rangle\}_{i=1,2,3}$, and let $i= 1$, $j = 2$ so we have 
positive projections $P_i = P_1: \rho \mapsto \pi\rho\pi$ with $\pi=|1\rangle\langle 1|$, as well as
$P_i'= P'_1: \rho \mapsto \pi'\rho\pi'$, where $\pi'=|2\rangle\langle 2|+|3\rangle\langle 3|$. The action of these on a $3\times 3$
density matrix $\rho$ is to set specific entries of the matrix to zero. 
More explicitly, $\ker P_1$ is the set of Hermitian matrices of the form 
$\left(\begin{array}{ccc} 0 & \bullet & \bullet \\ \bullet & \bullet & \bullet \\ \bullet & \bullet & \bullet
\end{array}\right)$, where $\bullet$ denotes an arbitrary entry. The $\bullet$'s correspond to the 
entries set to zero by $P_1$; interchanging the $\bullet$'s with the $0$'s would give the form of the matrices in $\im P_1$.  Similarly $\ker P_1'$ is the set of Hermitian 
matrices of the form 
$\left(\begin{array}{ccc} \bullet & \bullet & \bullet \\ \bullet & 0 & 0 \\ \bullet & 0 & 0
\end{array}\right)$. 
So $\ker P_1 \cap  \ker P_1'$ is the
set of all Hermitian matrices of the form
$\left(\begin{array}{ccc} 0 & \bullet & \bullet \\ \bullet & 0 & 0 \\ \bullet & 0 & 0
\end{array}\right)$. Since $P_j = P_2$  analogously projects onto
the span of $|2\rangle$, then $\im P_{ij} \equiv \im P_{12}$ is the set of Hermitian matrices of the form
$\left(\begin{array}{ccc} \bullet & \bullet & 0 \\ \bullet & \bullet & 0 \\ 0 & 0 & 0
\end{array}\right)$, and the intersection $\ker P_1 \cap  \ker P_1' \intersect
\im_+ P_{12}$ yields the offdiagonal elements corresponding to the two-slit experiment $P_i$, $P_j$ as claimed, i.e. the Hermitian matrices of the form
$\left(\begin{array}{ccc} 0 & \bullet & 0 \\ \bullet & 0 & 0 \\ 0 & 0 & 0
\end{array}\right)$.

The decomposition of Proposition \ref{prop: M slit
    Peirce decomposition} is thus into the spans of the faces $\im_+
  P_i$, $M(M-1)$ spaces associated with interference between these
  faces, and a further space, which as the next Proposition shows, is
  associated with three-way interference.

Proposition~\ref{prop: M slit Peirce decomposition} is stated as a decomposition of the vector space $A$.
However, note that every face of $A_+$ (with group of reversible transformations given by the restriction of those
global reversible transformations that preserve that face) is itself a state space satisfying Postulates 1 and 2.
Thus, if we have an incomplete $M$-slit mask with $F:=\im P_{12\dots M}$ and corresponding face $F_+:=F\cap A_+$,
we obtain a decomposition
\begin{equation}
   F=\oplus_i F_i \oplus_{i<j} F_{ij}^c\oplus F^{(3)},
   \label{eqPierceIncompleteMask}
\end{equation}
where $F_i=\im_+ P_i\subseteq F$, $F_{ij}^c=\ker P_i \cap \ker P_j\cap \im_+ P_{ij}\subseteq F$, and $F^{(3)}=\bigcap_{i<j} \ker P_{ij}\cap F$.
This is used in the following proposition.

\begin{proposition}
\label{Prop21}
Let $A$ be a state space satisfying Postulates 1 and 2.
Then there is no third-order interference on $A$ if and only if for every $M$-slit
mask $P_J$, $J\subset\{1,\ldots,M\}$ with $M\geq 2$, and every pure state $\omega\in \im P_{12\dots M}$,
the component $\omega^{(3)}$ of $\omega$ in $F^{(3)}$ in~(\ref{eqPierceIncompleteMask}) is zero.
\end{proposition}

\begin{proof}
From Lemma~\ref{LemEquivThirdOrder}, the absence of third order interference is equivalent to
\begin{equation}
   P_{12\cdots M} x = \sum_{i < j} P_{ij} x  - (M-2) \sum_i P_i x
   \label{eqAbsence3rdOrderEquiv}
\end{equation}
for all $x\in A$. However, since $P_{ij}=P_{ij} P_{12\cdots M}$ and $P_i = P_i P_{12\cdots M}$, this
is equivalent to~(\ref{eqAbsence3rdOrderEquiv}) holding for all $x\in\im P_{12\cdots M}=:F$. Since the pure states in $F$
span $F$, this is equivalent to
\[
   \omega=\sum_{i < j} P_{ij} \omega  - (M-2) \sum_i P_i \omega
\]
for all pure states $\omega\in F$. By Proposition~\ref{prop: M slit Peirce decomposition} and its consequence~(\ref{eqPierceIncompleteMask}),
$P_{ij} \omega = P_i \omega + P_j \omega + \omega^c_{ij}$, where $\omega^c_{ij}$ is the component of $\omega$ in $F_{ij}^c$.
So absence of third-order interference is equivalent to
\[
   \omega = \sum_{i < j} \left[ P_i \omega + P_j \omega + \omega^c_{ij} \right] - (M-2) \sum_i P_i \omega
\]
for all pure states $\omega\in F$. Noting that $\sum_{i < j} (P_i + P_j)$ contains, for each fixed value of $k$,
$M-1$ occurrences of $P_k$, this becomes:
\begin{eqnarray*}
\omega &=& (M-1) \sum_i P_i \omega + \sum_{i < j} \omega^c_{ij}
- (M-2) \sum_i P_i \omega \\  &=&
\sum_i P_i \omega + \sum_{i < j} \omega^c_{ij}\;.
\end{eqnarray*}
In other words, $\omega^{(3)} = 0$ in $F^{(3)}$ in~(\ref{eqPierceIncompleteMask}).
\end{proof}

\begin{definition}
The \emph{impurity} $I(\omega)$ of any unnormalized state $\omega\geq 0$
is defined as:
\[
   I(\omega):=u(\omega)^2 - \|\omega\|^2.
\]
\end{definition}
For normalized states $\omega\in\Omega$, we have $u(\omega)=1$, and $\|\omega\|\leq 1$, with equality
if and only if $\omega$ is a pure state. Extending this to the unnormalized states by multiplication $\omega\mapsto \lambda\omega$
with $\lambda\geq 0$ shows that $I(\omega)\geq 0$ for all $\omega\geq 0$, with equality if and only if $\omega$ is ray-extremal.

\begin{proposition}\label{prop: impurity inequality}
Let $P_i$ with $i \in \{1,...,M\}$ be an $M$-slit mask
on a system satisfying Postulates 1 and 2. Then for any state $\omega\in F:=\im P_{12\dots M}$ (not necessarily pure or normalized)
\beq
\label{eqImpurityInequality}
||\omega^{(3)}||^2
= \sum_{i < j} I(P_{ij} \omega) - I(\omega) - (M - 2)\sum_i I(P_i \omega)\;,
\eeq
where $\omega^{(3)}$ is the component of $\omega$ in $F^{(3)}$ in~(\ref{eqPierceIncompleteMask}).
\end{proposition}

While we use this equation directly in what follows, its
significance is underlined by noting its immediate corollary: that if
$\omega$ and each of the $P_i \omega$ are pure, and there is no
third-order interference, then (by the nonnegativity of impurity) each
of the $P_{ij} \omega$ is also pure.  In other words: in the absence
of third-order interference, if the $P_i$ are each purity-preserving,
so also are the $P_{ij}$.

\noindent
\emph{Proof of Proposition~\ref{prop: impurity inequality}:}
First we expand $\omega\in F$ via~(\ref{eqPierceIncompleteMask}):
\[
   \omega = \sum_i P_i \omega + \sum_{i < j} \omega^c_{ij} + \omega^{(3)}.
\]
Taking squared norms and using orthogonality of the decomposition, we get
\beq \label{norm expansion N slits}
||\omega^{(3)}||^2 =
||\omega||^2 - \sum_i ||P_i \omega ||^2 - \sum_{i < j} || \omega^c_{ij}||^2 \;.
\eeq
In order to get results about
the purity of $P_i \omega$ and $P_{ij} \omega$, we use
$P_{ij} \omega  = P_i  \omega + P_j \omega + \omega^c_{ij}$ to
eliminate $\omega^c_{ij}$ by substituting $||\omega^c_{ij}||^2 =
||P_{ij} \omega||^2 -  ||P_i \omega||^2 - ||P_j \omega||^2$ in~(\ref{norm expansion N slits}), obtaining:
\beqa
||\omega^{(3)}||^2 = \nonumber \\ \nonumber
||\omega||^2 - \sum_i ||P_i \omega ||^2 + \sum_{i < j}
\left[ ||P_i \omega||^2 + ||P_j \omega ||^2 - ||P_{ij} \omega ||^2 \right].
\eeqa
Since a given $k$ appears (as $i$ or $j$)
in $M-1$ of the pairs $i < j$, and the last sum in the above expression has a
$||P_k \omega||$ for each such appearance, this becomes
\[
||\omega^{(3)}||^2 =
||\omega||^2 + (M - 2) \sum_i ||P_i \omega ||^2
- \sum_{i < j} ||P_{ij} \omega ||^2 \;.
\]
Note that $||P_{ij}\omega||^2 = u(P_{ij} \omega)^2 - I(P_{ij}\omega)$, so
\beqa\label{an equation}
||\omega^{(3)}||^2 = ||\omega||^2  + (M - 2)\sum_i \|P_i \omega\|^2
\nonumber \\ + \sum_{i < j} I(P_{ij}\omega) -  \sum_{i<j} u_{ij}(\omega)^2 .
\eeqa
Using $u_{ij} = u_i + u_j$,
\[
   \sum_{i<j} u_{ij}(\omega)^2 =
\sum_{i < j} \left(u_i(\omega)^2 + u_j(\omega)^2
+ 2 u_i(\omega) u_j(\omega) \right)\;.
\]
Again using the fact that
a given $i$ appears in $M-1$ of the pairs $i < j$, and
writing $\sum_{i} \sum_{j \ne i}$ in place of $2 \sum_{i < j}$, this becomes:
\begin{eqnarray*}
(M - 1) \sum_i u_i(\omega)^2  + \sum_{i} \sum_{j \ne i} u_i(\omega) u_j(\omega)\\
= (M - 1) \sum_i u_i(\omega)^2  + \sum_{i}\left[ u_i(\omega) \sum_{j \ne i}
u_j(\omega)\right]\;.
\end{eqnarray*}
Now, since $u_F(\omega)=\langle P_F u,\omega\rangle=\langle u,P_F \omega\rangle=u(\omega)$
and $\sum_{j \ne i} u_j(\omega) = u_F(\omega) - u_i(\omega)$, we get
\begin{eqnarray*}
\sum_{i<j} u_{ij}(\omega)^2 &=& (M - 1) \sum_i u_i(\omega)^2  + \sum_{i} u_i(\omega) \left(\strut u_F(\omega) - u_i(\omega)\right) \\
&=& u_F(\omega)^2  + (M - 2) \sum_i u_i(\omega)^2 \\
&=& u(\omega)^2 + (M - 2) \sum_i u_i(\omega)^2.
\end{eqnarray*}
Substituting this into~(\ref{an equation}) and rearranging gives~(\ref{eqImpurityInequality}).
\qed
\\

We will use this result several times in an inductive argument to
establish that all filters are purity preserving.

\begin{proposition}\label{interference and purity preservation}
Let a system $A$ satisfy Postulates 1 and 2. Then it has no third-order interference
if and only if all its filters are purity-preserving.
\end{proposition}

\begin{proof}
Suppose that all filters are purity-preserving. Then, if $P_i$, $i\in\{1,\ldots,M\}$ is any $M$-slit mask and $\omega$ is a pure state in $\im P_{12\dots M}$,
we have $I(\omega)=I(P_i\omega)=I(P_{ij}\omega)=0$, and so~(\ref{eqImpurityInequality}) implies that the component $\omega^{(3)}$ of $\omega$
in $F^{(3)}$ in~(\ref{eqPierceIncompleteMask}) is zero. Then Proposition~\ref{Prop21} implies that there is no third-order interference.

To show the converse direction, note first that it follows from~\cite[Prop.\ 7.28]{ASBook} in the context of Postulates 1 and 2 that all filters are of the form $P_F$
for some face $F$; thus, we only have to show that these orthogonal projections are purity-preserving.
Let $N$ be the size of $A$'s largest frame.
The proof that all filters are purity-preserving will be inductive on
the rank of filters.  The base case is rank-1 filters, which holds
because a rank-one filter projects the state onto the span of an extremal
ray of $A_+$.

We now prove the induction step, which states that if for some fixed rank $k\leq N-1$ all filters are purity-preserving,
then all filters of rank $k+1$ are purity preserving.
Suppose filters of rank $k$ are purity-preserving and
consider any mask consisting of a rank-$k$ filter
$P_1$ and $N-k$ rank-1 filters $P_i, i \in \{2,...,N-k+1\}$.
Then for any pure state $\omega$ each $P_i \omega$ is pure.  So with $||\omega^{(3)}||^2=0$
by the absence of third-order interference,~(\ref{eqImpurityInequality}) becomes:
\[
   \sum_{ij} I(P_{ij}\omega) = 0 \;.
\]
Since impurity is nonnegative, each of the $P_{ij} \omega$ is pure too.
So all the $P_i \join P_j$, and in particular the rank-$(k+1)$ filters
$P_1 \join P_i$, $i \in \{2,...,N-k+1\}$, are purity-preserving.
Since every rank-$(k+1)$ filter on $A$ has the form
$P \join Q$ for some rank-$k$
$P$ and some rank-1 $Q$ orthogonal
to $P$, all rank-$(k+1)$ filters
on $A$ are purity-preserving, and the induction step is established
for $k \leq N-1$.  Hence all filters of rank
up to $N-1$ are purity-preserving.
\end{proof}

In the context of assumptions (a) and (b) of Theorem~\ref{ThmAlfsenShultz}, Postulate 3${}'$ is known~\cite{ASBook, GuzFilter78}
to be equivalent to another postulate: that the lattice of exposed faces
has the~\emph{covering property}. We say that an element $F$ of a lattice \emph{covers} another element $G$ if
$G$ is below $F$ and there is nothing between them. Hence an atom is an element that covers $0$.
By definition, a lattice has the covering property if for every element $F$ and atom $a$, either $F\vee a = F$
or $F\vee a$ covers $F$.
  
In the context of Postulates 1 and 2,
the covering property can be formulated as follows: if $F$ is any face of $A_+$, and $\omega$ a
  pure state, then the face $G$ generated by both has rank
  $|G|\leq|F|+1$.  Since we have shown that (a) and (b) of Theorem
  \ref{ThmAlfsenShultz} follow from Postulates 1 and 2, the covering property
  can replace the absence of third order interference (or Postulate
  3$^\prime$).

\section{Standard quantum theory from observability of energy}

In standard quantum mechanics, we are used to treating the generator
of time evolution as an observable: evolution of any closed quantum
system with initial state $\rho_0$ is given by
\[
   \rho(t)=e^{-i H t} \rho_0 e^{i H t},
\]
where $H=H^\dagger$ is the system's Hamiltonian. The right-hand side,
as a one-parameter group acting on $\rho$, is generated by the
superoperator
$X: \rho \mapsto -i[H,\rho]$,
so that
$\rho(t)=e^{tX}\rho_0$. We are used to associating the observable
$E:  \rho \mapsto \tr (H \rho)$ with this generator, and call it the
``expectation value of energy''.

It is an interesting question why such an association is possible -- what is the operational relation between $E$ and $X$? The following properties
characterize this relation:
\begin{itemize}
\item If $X$ and $X'$ are two different generators, then the corresponding observables satisfy $E'\neq E$. That is, the observable determines
the generator uniquely.
\item The observable $E$ is a conserved quantity of the time evolution generated by $X$: $E(\rho(t))=E(\rho_0)$.
\item If time evolution is not trivial (i.e.\ $\rho(t)$ not constant), then $E$ is also not a trivial observable:
there are at least two states $\rho,\sigma$ such that $E(\rho)\neq E(\sigma)$.
\item The map $X\mapsto E$ is linear -- in particular, larger values of $E$ correspond to ``faster'' time evolution.
\end{itemize}

These properties allow us to define a notion of ``observability of energy'' for arbitrary probabilistic theories, which will turn
out to be a rather restrictive property.
\begin{definition}
\label{DefEnergyMeasurement}
Let $A$ be any state space with a group of reversible transformations $\G_A$.
An \emph{energy observable assignment} is an injective linear map $\phi:\g_A\to A^*$ such that the observable $\phi(X)$ is
conserved under the time evolution generated by $X$, but not under \emph{all} time evolutions unless $X=\phi(X)=0$.
We say that ``energy is observable'' on system $A$ if $\g_A\neq\{0\}$ and if there exists an energy observable assignment.
\end{definition}

Our fourth postulate is thus
\begin{postulate}Energy is observable on every system.
\end{postulate}

Writing the time evolution starting with initial state $\omega_0$ explicitly as
\[
   \omega(t):=e^{tX}\omega_0,
\]
a conserved quantity $E\in A^*$ is a linear functional with $E(\omega(t))=E(\omega_0)$. It is easy to check that this
is equivalent to $E\circ X=0$, where ``$\circ$'' is for composition of linear maps. If $E$ were equal to the order unit,
i.e.\ $E=u_A$, then $E(\omega(t))=u_A(\omega(t))=1$ for all $t$ and \emph{all} time evolutions, since all elements of
$\G_A$ preserve the normalization. Thus, Definition 28 implies the conditions
\begin{eqnarray*}
    \phi(X)\circ X &=& 0 \mbox{ for all }X\in\g_A, \\
    u_A&\not\in&{\rm ran}(\phi).
\end{eqnarray*}
Our notion is related to Alfsen and Shultz's notion of a ``dynamical
correspondence''~\cite{ASBook}, except that they require an injection
of observables into dynamical generators, rather than vice versa, and
in addition to a conservation condition, impose a 
condition relating reversible transformations to general automorphisms
of the cone of states which is formulated in the Jordan-algebraic
setting.  Our setting is more general, and we impose no such relation
between the reversible transformations and cone automorphisms.
Connes~\cite{Connes} used a notion of \emph{orientation} related to
dynamical correspondence to characterize the state spaces of
von Neumann algebras (one of the infinite-dimensional generalizations
of standard quantum systems) among those of JBW-algebras (one infinite-dimensional generalization of finite-dimensional formally real Jordan algebras).
Other work making use of similar notions to characterize quantum and
classical theory in different settings, can be found in 
\cite{GrginPetersen1, GrginPetersenComposability, Kapustin, Moldoveanu}.  
References~\cite{Kapustin}, \cite{Moldoveanu}, and~\cite{GrginPetersen1} all \emph{derive} relations
between energy and observables, and thence that the theory must
essentially be standard quantum or classical, from considerations involving dynamics on composites, so our
work is complementary to theirs in that we avoid assumptions about composite
systems.

The identification of dynamical generators with conserved observables
that exists in classical and quantum theories is central to many
physical phenomena and arguments, providing motivation for our
postulate.  We mention in particular that standard formulations of the
statistical mechanics underlying thermodynamics use a conserved energy
observable in the definition of free energy.

Our goal is to show the following:

\begin{theorem}
Postulates 1, 2, 3, and 4 imply that the state space is an $N$-level state space of standard complex quantum theory, for some $N\in\N$,
and all conjugations $\rho\mapsto U\rho U^\dagger$ with $U\in {\rm SU}(N)$ are contained in the group of reversible transformations.
\end{theorem}
\begin{proof}
We show that complex quantum $N$-level state spaces are the only finite-dimensional irreducible formally real Jordan
algebra state spaces that have observability of energy. This is enough due to Theorem~\ref{TheoremInterference}.

First, consider the $d$-dimensional ball state spaces (``spin factors'')
\[
   \Omega_d:=\{(1,r)^T\,\,|\,\ r\in\R^d,\enspace \|r\|\leq 1\}.
\]
(The qubit appears again in this class of systems, since the $d=3$ case is the Bloch ball.)
The Lie algebra is non-trivial only for $d\geq 2$. Consider the case that the group $\G_d$ of reversible transformations contains
the full orthogonal group, such that $\g_d={\rm so}(d)$. If Postulate 4 holds, then there must be an injective linear map $\phi$ from
the Lie algebra $\g_d$ of $\G_d$ to $\R^{d+1}$. But $\dim({\rm so}(d))=d(d-1)/2$ which is larger than $d+1$ for $d\geq 4$,
so no such map can exist for $d\geq 4$. If $d=2$, we have
\[
   \g_2=\R\cdot\left(
      \begin{array}{ccc}
         0 & 0 & 0 \\
         0 & 0 & 1 \\
         0 & -1 & 0
      \end{array}
   \right),
\]
and calling this matrix $X$, it is easy to see that $\phi(X)\circ X=0$ implies that $\phi(X)=c\cdot u_d$ for the normalization functional
$u_d(x_1,x_2,x_3)=x_1$. This contradicts the definition of an energy observable assignment.

If $d$ is even or $d=7$, there are compact connected subgroups of ${\rm SO}(d)$ that are transitive on the pure states of $\Omega_d$,
and thus satisfy Postulate 2, (see Ref.~\cite{MasanMulEntang} for the list of groups; they have been classified in~\cite{MontgomerySamelson,Borel}).
As we show in Appendix~\ref{AppendixB}, all of these cases except for one can be ruled out by dimension counting, exactly
as the cases $d\geq 4$ above; the only case where this does not work is $d=4$ with transformation group $\G_2={\rm SU}(2)$.
But there, it can be shown that there are time evolutions which only have the normalization as their conserved observable,
contradicting Definition~\ref{DefEnergyMeasurement}.

Now let $A$ be the state space of the $3\times 3$ octonionic matrices. Due to Postulate 2 (in the special case of $1$-frames), the group of reversible transformations $\G_A$
acts transitively on the pure state manifold, which is the Cayley plane $P^2(\mathbb{O})$, hence so does its connected component
at the identity~\cite{MMAxioms}. According to~\cite{Tsukuda} and~\cite{Onishchik}, the only compact connected
Lie group which acts transitively and effectively on it is the exceptional Lie group ${\rm F}_4$. But $\dim({\rm F}_4)=52 > \dim(A)=27$, so there is no
injective linear map from $\g_A$ to $A^*$.

For $N\geq 3$, consider the $N$-level state space $A_N$ of quaternionic quantum mechanics, with any group of reversible transformations $\G_N$ satisfying
Postulate 2. Then $\dim(A_N)=2 N^2-N$. The pure states define the quaternionic projective space $P^{N-1}(\mathbb{H})$, and so $\G_N$
must act transitively on it. According to~\cite{Onishchik}, the only possibility is $\g_N\supseteq {\rm sp}(N)$, and $\dim({\rm sp}(N))=N(2N+1)$, which is larger than $\dim(A_N)$.

The only remaining cases are the $N$-level state spaces $A_N$ of real quantum mechanics for $N\geq 3$, which are more difficult to rule out -- dimension counting
does not work. First, it can be shown from the classification results of~\cite{Tsukuda} that Postulate 2 implies that the group of reversible transformations
contains all maps of the form $\rho\mapsto O\rho O^T$ with $O\in {\rm SO}(N)$; consequently, every map $X(\rho):=\rho\mapsto [M,\rho]$ with $M\in {\rm so}(N)$ is a valid generator.
An energy observable assignment $\phi$ maps these generators (resp.\ the matrices $M$) to observables (that is, symmetric matrices $\bar M$) such that
$[\phi(X)](\rho)=\tr(\bar M \rho)$; the conservation condition $\phi(X)\circ X=0$ becomes $[M,\bar M]=0$. However, as we show in the appendix by considering
certain special generators $X$, all maps of this kind must have $\bar M=\mathbf{1}$ in their range, yielding the normalization functional, which contradicts
the definition of an energy observable assignment.

In the standard case of complex $N$-level quantum theory, it remains to show that the group of reversible transformations $\G_N$ contains all unitaries
(it might also contain anti-unitaries; due to Wigner's theorem~\cite{Bargmann,Wigner}, these are the only possibilities).
Postulate 2 implies transitivity of the connected subgroup of $\G_N$ on the pure states, hence on the projective space $P^{N-1}(\C)$;
according to~\cite{Onishchik}, for odd $N$, the only possibility is the projective action of ${\rm SU}(N)$; but if $N$ is even, say $N=2n$, there is a second possibility,
which is the projective action of ${\rm Sp}(N)$.
But consider two $N$-frames $|e_1\rangle\langle e_1|,\ldots, |e_N\rangle,\langle e_N|$ and $|f_1\rangle\langle f_1|,\ldots,|f_N\rangle,\langle f_N|$,
where $e_1,\ldots, e_N$ are defining vectors of the basis in which $J=\left(\begin{array}{cc} 0 & -\mathbf{1} \\ \mathbf{1} & 0\end{array}\right)$,
such that a unitary $U$ is in ${\rm Sp}(N)$ if and only if $U^T J U=J$. Moreover, suppose that $f_1=e_1$. If Postulate 2 is satisfied, there is $U\in {\rm Sp}(N)$
such that $U |e_i\rangle\langle e_i| U^\dagger=|f_i\rangle\langle f_i|$ for all $i$, so $U e_1 = e^{i\varphi} e_1$ for some $\varphi\in \R$. Since
$J e_1=e_{n+1}$, it is easy to see that the symplectic constraint on $U$, together with $U^\dagger e_{n+1}=\overline{(U^T e_{n+1})}$,
implies that $U e_{n+1}=e^{-i\varphi} e_{n+1}$, so $|f_{n+1}\rangle\langle f_{n+1}|
=|e_{n+1}\rangle\langle e_{n+1}|$, which contradicts frame transitivity, i.e.\ Postulate 2.
\end{proof}

The fact that energy observability rules out classical systems in
this theorem is a consequence of our finite-dimensional setting,
for which classical reversible dynamics are a discrete group.  
The probabilistic representation of phase-space classical mechanics
involves an infinite-dimensional space of Liouville distributions, and
does, of course, have continuously parametrized reversible dynamics.

\section{Discussion and conclusions}

We have given four principles that we argue have, to various degrees,
the virtues of conceptual clarity, important physical implications,
intuitive appeal, and interesting experimental consequences.
We have
shown that while they are formulated in the setting of an extremely
broad class of probabilistically described systems together they
constrain the abstract structure of such a system to be that of the
usual Hilbert space quantum theory over the complex field.  Our
demonstration was limited to finite dimension, a limitation which we
believe to be primarily technical. This reconstruction of quantum theory differs
interestingly from several previous ones in avoiding any postulates
concerning the structure or even existence of composite systems.

Another desirable feature of our reconstruction is its stepwise
structure, in which conceptually and often physically
significant properties appear even as a consequence of the first
postulate, and additional such properties appear at each step.

Postulates 1 and 2 together further have very strong consequences:
they imply that all effects are allowed, that every face of the state space is the image of a
filter, i.e., that the state space is \emph{projective}, and also that
it is self-dual.  Filters allow one to verify that a state is in a
claimed face of the state space without (if the claim is true)
disturbing the state.  They are likely to be important ingredients of
both information-processing and thermodynamical protocols;
possibilities which are under investigation.  Filters can also be used
to equip a system with operations destroying coherence between any set
of mutually orthogonal faces. In other words, the existence of filters
ensures the possibility of a process of
\emph{decoherence} similar to the one in quantum theory.

Self-duality is another strong property of state spaces that is
independent of projectivity.
Self-duality introduces a correspondence between atomic measurement
outcomes and pure states that is exploited in quantum steering and
teleportation, for example.  It is
also known to be linked, in some special contexts such as polygonal
state spaces, to correlations satisfying the Tsirel'son bound on 
violations of Bell locality~\cite{JanottaEtAl}.

The lattice of faces given Postulates 1 and 2 is orthomodular---as is
implied, indeed, by projectivity.  This expresses a kind of ``local
classicality'', which one sees also in the topos-theoretic approach
of e.g. \cite{IshamButterfield2000}, and also
relates our work to the classic ``quantum logic'' approach initiated by
Birkhoff and von Neumann \cite{BirkvN36}.  Postulate 2 imposes a high degree of
symmetry on this lattice---it would be interesting to investigate
lattices with such high symmetry using purely lattice-theoretic
methods.

There is a close connection between Postulate 2 and certain properties
of the circuit model for quantum computation. In this model it is
standard to start with an input $n-$level system in a particular
state, as well as a number of other $n-$level systems which can
without loss of generality be taken to be in the $|0 \rangle$
state. Then we implement the circuit representing the computation we
wish to carry out, and at the end we must measure a specific observable
to determine the (probability of the) output of the computation. This
last measurement step can be done without loss of generality by first
reversibly transforming the (generally entangled) logical $n-$level
system of interest into an individual physical $n-$level system,
and then doing the desired measurement on this system alone. This transfer is
possible because quantum theory satisfies Postulate 2.  
Postulate 1 and 2 together can be understood as generalizing this idea
by demanding that every state (not just pure ones) of a system can be
transferred to any other system (with the same or larger number of
distinguishable states) by a suitable reversible interaction, provided
both are subsystems of a common larger system.

Our third postulate provides, in the context set by the first two
postulates, a perhaps surprising link between the absence of
irreducibly three-slit interference, currently under experimental
scrutiny, and mathematical notions: the Jordan algebraic structure of
quantum theory on the one hand, and the satisfaction of the covering
law by its lattice of faces, on the other.  In the context of our
first two postulates, these are all equivalent.  The known equivalence
(even in the broader context of projective systems) of the latter two
with the requirement that filters preserve purity is further food for
thought.  An interesting question is whether the equivalence of no
higher-order interference with either of
 these two principles still holds in the
broader projective context.  Looking to operational consequences,
perhaps the failure of purity preservation might give rise to an extra
source of noise or irreversibility in information processing or
thermodynamical protocols---though this might be circumvented if the
protocols are designed so the states being filtered are ``compatible''
with the filters.

Most interesting, perhaps, is the possibility that there exist
families of systems satisfying our first two postulates but not the
third: these would still have an extremely regular structure and
likely support interesting information processing, but so far no
examples are known. Should they be shown not to exist,
we would then know that Jordan systems are singled out by Postulates
1 and 2 alone.

The final step, narrowing things down from Jordan systems to complex
quantum systems via energy observability, is not so surprising.
Similar postulates have been used for this purpose by Connes and by
Alfsen and Shultz.  We require an injection of dynamical generators
into the space of observables, each injected generator conserved by
the dynamics it generates, whereas Alfsen and Shultz require the
converse and also impose ancillary conditions. In contrast, our
condition, though applied only to Jordan algebraic systems, is
formulated in greater generality where the ancillary conditions do not
make sense.  It is likely that in the Jordan-algebraic setting, the
ancillary conditions, as well as a bijection, are obtained
automatically.  Exploration of conditions of this type---either ours,
or abstractions of Connes' or Alfsen and Shultz's---in a broader
context are desirable.  Indeed, as we have mentioned, others have
explored similar principles, though some of these investigations 
(e.g. \cite{Kapustin}, \cite{Moldoveanu}) have
made use of composite systems which appear to us to be required to
satisfy local tomography.  In the context of our Postulates 1 and 2,
locally tomographic composites and the existence of stand-alone
$2$-level systems would imply that the systems are standard quantum
systems; indeed one reason for our interest in energy observability is
as an \emph{alternative} to local tomography.

The fact that energy observability rules out classical systems in this
theorem is an artifact of our finite-dimensional setting, for which
classical reversible dynamics are a discrete group.  Since
infinite-dimensional classical systems \emph{do} have continuous
one-parameter groups of reversible transformations, however, it is
important to point out that there are numerous alternative assumptions
which would allow us to rule out classical systems in the
finite-dimensional case without assuming the existence of continuous
reversible dynamics.  Such alternatives are likely to retain their
usefulness in infinite dimensions. For example, we could postulate
the existence of a tradeoff between information gained in a
measurement and disturbance to the measured state \cite{Barrett}, or
the existence of at least one state that has two distinct convex
decompositions into pure states, or the existence of interference; the
existence of nonclonable or nonbroadcastable sets of states
\cite{BBLW08clone, BBLW07broadcast} might also work.

Although we are not aware of work using the set of postulates we use,
several authors have used one or more related principles.  In Wilce's
characterization in \cite{Wilce4.5axioms}, a symmetry principle
reminiscent of our Postulate 2 (but concerning test spaces rather than
state spaces) was used, along with reversible transitivity on pure
states (a special case of Postulate 2). In his most recent
reconstruction, Hardy~\cite{Hardy2011} uses a postulate (``filters are
non-flattening'') which relies on a definition of filters that is
equivalent to ours (at least in the context of our Postulates 1 and
2), and which implies Postulate 3${}'$ (that filters are
purity-preserving).  Niestegge has also used the absence of
higher-order interference as one ingredient in deriving Jordan
algebraic systems \cite{NiesteggeInterference}.  In
\cite{InterferenceInPrep} it was established that finite-dimensional
Jordan systems do not have higher-order interference, a result also
found by Niestegge in \cite{NiesteggeInterference}.

Daki\'c and Brukner~\cite{DakicBrukner} have used Postulate 1 and the fact that all
pure states are connected by reversible transformations to derive the ball shape of two-level
state spaces (a fact that carries over to all two-level systems satisfying our Postulates 1 and 2).
In their reconstruction of quantum theory, Chiribella, d'Ariano, and Perinotti~\cite{GiulioDer11} have proven
several lemmas that are close to some of ours (such as statements on positive projections, or a sum representation
of projective units), but obtained them from different assumptions.
We have already
mentioned other work postulating connections between observables and
dynamical generators.  More work understanding the connections between
the various approaches would likely be fruitful.

Besides providing an understanding of the Hilbert space structure of quantum theory from first principles, our reconstruction suggests a variety of open questions, such as the existence of systems with strong symmetry and classical decomposability, but also with higher-order interference. Furthermore, we think that the naturalness of our postulates allows us to make closer contact with other aspects of physics, a direction we consider important to pursue.

This is evident from the postulates themselves -- Postulate 3 considers a property that is under direct experimental investigation, and so solving the aforementioned open problem might provide concrete consistent models that can be tested against quantum theory in experiments. Postulate 4 relates the probabilistic structure to the existence of a notion of energy of the form physicists are used to. Furthermore, consequences of the postulates -- such as projectivity -- seem crucial for thermodynamic reasoning. In fact, weaker versions of Postulates 1 and 2, in conjunction with local tomography, are enough to make sense of the general-probabilistic thermodynamics results
in~\cite{GPTThermo1,GPTThermo2}.

In this sense, our result is part of a broader research program: analyze the structure of
physics -- that is, the way that the different parts of physics fit
together -- by rigorously assessing the consequences of changing some
of its parts.  One part of physics is quantum theory, and seeing how a
more general probabilistic theory could still harmonize with
thermodynamics or Hamiltonian mechanics is one of many ways to gain
insights into the way our world works.  Given the current quest for a
theory that unifies quantum and gravitational physics, in a situation
where conclusive experimental results are mostly absent, it seems
particularly promising to rigorously analyze the logical and
conceptual structure of what is known, hoping thereby to glimpse a
path towards the unknown.

\section*{Acknowledgments}
Some of this work was done while the authors were employed by or visiting the Perimeter Institute for Theoretical Physics, Waterloo, Ontario, Canada.
Research at Perimeter Institute is supported in part by the Government
of Canada through NSERC and by the Province of Ontario through MRI.
Also, part of the work was done while HB was a Fellow of the Stellenbosch
Institute for Advanced Studies at the Wallenberg Research Center at 
Stellenbosch University, in 2012. Furthermore, we would like to thank two anonymous
referees for the thorough review of the manuscript and for helpful suggestions, and one in particular for pointing out a mistake in an earlier
version of Proposition~\ref{SpectralityImpliesSaturation} and for suggesting how to fix it.

\onecolumngrid

\begin{appendix}

\section{Perfection and positive projections}
\label{AppendixIochum}
In this section, we give a proof of the following proposition which is originally due to Iochum~\cite{IochumThesis,IochumBook}.
\begin{proposition}
Let $A_+$ be a regular self-dual cone in $A$.  $A_+$ is 
perfect if and only if each orthogonal (with 
respect to the self-dualizing inner product) projection $P_F$ 
onto the linear span $F$ of a face $F_+$, is positive.
\end{proposition}

\begin{proof} 
We write $F_+^*$ for the dual of $F_+$ in $F$, according to the 
restriction of the self-dualizing inner product for $A_+$; thus 
perfection means that $F_+^* = F_+$ for every face. 

We begin with ``only if''.  Let $P$ be the orthogonal projector onto
$F$, $x \in A_+$, $y \in F_+$.
Now $\langle y, Px \rangle = \langle P^* y, x \rangle$; since $P$ is
Hermitian this equals $\langle Py , x \rangle = \langle y , x
\rangle$.  The latter is nonnegative because both $y$ and $x$ are in
$A_+$, which is self-dual.  So we have shown 
$\forall y \in F_+\, \langle y, Px \rangle \ge 0$, i.e. $Px \in F_+^*$.  
But by perfection $F_+^* = F_+$.  Thus $Px \in F_+$ for any $x \in A_+$, 
i.e. $P$ is positive. 

For ``if'', we begin by observing that given positivity of $P$,
$P A_+ = F_+$.   This is because $P x = x$ for any $x \in F$, 
so $P F_+ = F_+$, whence $P A_+ \supseteq F_+$; 
on the other hand $P A_+ \subseteq F_+$ by positivity.  

Note that $F_+ \subseteq F_+^*$ as a consequence of self-duality of 
$A_+$:  since everything in $F_+$ is in $A_+$, it must have nonnegative
inner product with everything in $A_+$, hence with everything in $F_+$, 
and since it is in addition in $F$, it is in $F_+^*$.   Recall
that $y \in F_+^*$ is defined as 
 $y \in F$ and satisfying $\forall x \in F_+ \,
\langle y, x \rangle \ge 0$.  Since $PA_+ = F_+$, the latter part
of this condition is equivalent to
$\forall z \in A_+\, \langle y, Pz \rangle \ge 0$.  Again moving the projector
to act on $y$, using its Hermiticity and that $y \in F$ so $Py=y$, this 
is equivalent to $\forall z \in A_+\, \langle y, z \rangle \ge 0$, i.e. 
$y \in A_+^*$.  Since  $A_+^* = A_+$ and $y$ was also assumed in $F$, 
$y \in F_+$, establishing that $F_+^* \subseteq F_+$.  We have now shown
$F_+^* = F_+$, i.e. perfection.
\end{proof}

\section{Calculations for observability of energy}
\label{AppendixB}
The goal of this section is to show the following:
\begin{lemma}\label{LemGroups}
The possible state spaces satisfying Postulates 1, 2 and 3 which have a non-trivial connected component $\G_0$ of
their reversible transformation groups are the following:
\begin{itemize}
\item The $d$-dimensional ball state spaces $\Omega_d:=\{(1,r)^T\,\,|\,\ r\in\R^d,\enspace \|r\|\leq 1\}$ with $d\geq 2$, and either $\G_0={\rm SO}(d)$, or
$\G_0={\rm SU}(d/2)$ if $d=4,6,8,\ldots$, or $\G_0={\rm U}(d/2)$ if $d=2,4,6,8,\ldots$, or $\G_0={\rm Sp}(d/4)$ if $d=8,12,16,\ldots$, or $\G_0={\rm Sp}(d/4)\times {\rm U}(1)$ if
$d=8,12,16,\ldots$, or $\G_0={\rm Sp}(d/4)\times {\rm SU}(2)$ if $d=4,8,12,\ldots$, or $\G_0={\rm G}_2$ if $d=7$, or $\G_0={\rm Spin}(7)$ if $d=8$, or $\G_0={\rm Spin}(9)$ if
$d=16$,
\item $N$-level real quantum theory with $N\geq 2$ and $\G_0=\{\rho\mapsto O \rho O^T\,\,|\,\, O\in {\rm SO}(N)\}$,
\item $N$-level complex quantum theory with $N\geq 2$ and $\G_0=\{\rho\mapsto U \rho U^\dagger\,\,|\,\, U\in {\rm SU}(N)\}$,
\item $N$-level quaternionic quantum theory with $N\geq 2$ and $\G_0\simeq {\rm Sp}(N)/\{-\mathbf{1},+\mathbf{1}\}$ (see~\cite{Bargmann,Onishchik}),
\item $3$-level octonionic quantum theory with $\G_0\simeq {\rm F}_4$.
\end{itemize}
However, among those, only the complex quantum theory state spaces (including $\Omega_3$, the qubit) satisfy Postulate 4, that is, observability of energy.
\end{lemma}

In complex quantum theory, the group of reversible transformations
$\G$ can actually be larger: it may also contain the antiunitary
transformations according to Wigner's theorem (but not
more). Similarly, real quantum theory may also contain the
conjugations with $O\in {\rm O}(N)$ (which yields additional transformations for even $N$), but for quaternionic quantum theory with $N\geq 3$, we
have $\G=\G_0$~\cite{Adler}. As pointed out in~\cite{Adler,Bargmann}, the case $N=2$ is exceptional in
the quaternionic case. Since the state space is in this case a $5$-dimensional unit ball, $\G$ may
contain reflections in adition to the rotations ${\rm SO}(5)\simeq {\rm Sp}(2)/\{-\mathbf{1},+\mathbf{1}\}$.
We do not know whether octonionic $3\times 3$ quantum theory may contain additional elements in its
transformation group, and we do not know the complete classifications
of possible compact transformation groups $\G\supset\G_0$ for the ball
state spaces (except that obviously $\G\subset {\rm O}(d)$).

Lemma~\ref{LemGroups} will be proven step by step.
We start by showing that the only ball state space with transitive
group of reversible transformations that has observability of energy
is the qubit.
\begin{lemma}
For $d\geq 2$, consider the $d$-dimensional ball state space
$\Omega_d:=\{(1,r)^T\,\,|\,\ r\in\R^d,\enspace \|r\|\leq 1\}$,
and let $\G_d$ be any compact group of reversible transformations
that acts transitively on the pure states. Then energy is observable
(in the sense of Definition~\ref{DefEnergyMeasurement})
if and only if $d=3$.
\end{lemma}
\begin{proof}
If $\G_d$ acts transitively on the pure states, then so does
its connected component at the identity~\cite{MMAxioms}. According
to~\cite{MasanMulEntang}, the list of groups is the following. Since
the group action is locally effective~\cite{Onishchik}, the dimensions
of $\g_d$ are just the dimensions of the corresponding groups.
\begin{itemize}
\item For all $d\geq 2$: ${\rm SO}(d)$. We have shown in the main text that
  an energy observable assignment only exists if $d=3$.
\item For $d=4,6,8,\ldots:$ ${\rm SU}(d/2)$. We have
  $\dim({\rm su}(d/2))=(d/2)^2-1$, and this is larger than $d+1$ if $d\geq
  6$. Thus, no injective map $\phi:{\rm su}(d/2)\to\R^{d+1}$ defining an
  energy observable assignment can exist. However, we have to treat $d=4$
  separately. In this case, the transformation group is (up to
  similarity)
\[
   \G_4\supseteq \left\{
      \left.
      \left(
         \begin{array}{ccc}
            1 & 0 & 0 \\
            0 & {\rm re}\, U & {\rm im}\, U \\
            0 & -{\rm im}\, U & {\rm re}\, U
         \end{array}
      \right)\,\,\right|\,\, U\in {\rm SU}(2)
   \right\},
\]
such that the Lie algebra is at least
\[
   \g_4\supseteq \left\{
      \left.
      \left(
         \begin{array}{ccccc}
              0 & 0 & 0 & 0 & 0 \\
              0 & 0 & b & a & c \\
              0 & -b & 0 & c & -a \\
              0 & -a & -c & 0 & b \\
              0 & -c & a & -b & 0
         \end{array}
      \right)\,\,
      \right|\,\, a,b,c\in\R
   \right\}.
\]
Let $X\in \g_4$ be a generator corresponding to the choice of
parameters $a=1$ and $b=c=0$. If $\phi$ is any energy observable assignment, we
can write the functional $\phi(X)$ as a vector $\varphi\in\R^5$ such
that $[\phi(X)](y)=\langle \varphi,y\rangle$ for all $y\in\R^5$, and
the condition $\phi(X)\circ X=0$ translates into $X^T \varphi=0$.
The kernel of $X^T$ is one-dimensional, with unique solution (up to
some factor) of $\varphi=\lambda\cdot (1,0,0,0,0)^T$,
$\lambda\in\R$. But this represents the normalization functional:
$\langle\varphi,y\rangle=u_4(y)$ for all $y$, so $\phi(X)=u_4$,
contradicting the definition of an energy observable assignment.
\item For $d=2,4,6,8,\ldots$: ${\rm U}(d/2)$. The case $d=2$ is already
  covered in the main text; in all other cases, this representation
  contains the corresponding representation of ${\rm SU}(d/2)$ as a
  subgroup, and this has already been treated.
\item For $d=8,12,16,\ldots$: ${\rm Sp}(d/4)$. Dimension counting rules out
  these cases: We have $\dim({\rm sp}(d/4))=d/4(2\cdot d/4+1)$, and this is
  larger than $d+1$ for the relevant dimensions.
\item For $d=8,12,16,\ldots$: ${\rm Sp}(d/4)\times {\rm U}(1)$. This
  representation contains the representation of ${\rm Sp}(d/4)$ as a
  subgroup; thus, it is ruled out by the previous case.
\item For $d=4,8,12,\ldots$: ${\rm Sp}(d/4)\times {\rm SU}(2)$. If $d\geq 8$, this too
  contains ${\rm Sp}(d/4)$ as a subgroup. If $d=4$ then the dimension of the group ${\rm Sp}(1)\times {\rm SU}(2)$ is $9$,
  which is larger than $d+1=5$.
\item For $d=7$: the exceptional Lie group ${\rm G}_2$. Dimension counting
  again: $\dim {\rm g}_2=14 > 7+1$.
\item For $d=8$: ${\rm Spin}(7)$. $\dim {\rm Spin}(7)=7(7-1)/2 = 21 > 8+1$.
\item For $d=16$: ${\rm Spin}(9)$. $\dim {\rm Spin}(9)=9(9-1)/2 = 36 > 16+1$.
\end{itemize}
This proves the claim.
\end{proof}

As mentioned in the main text, it is more difficult to rule out
$N$-level real quantum mechanics for $N\geq 3$. This needs a sequence
of lemmas.
\begin{lemma}
\label{LemLinAlg1}
Let $J=\left(\begin{array}{cc} 0 & -1 \\ 1 & 0 \end{array}\right)\in
\R^{2\times 2}$, and let $S\in\R^{2\times 2}$ such that
\[
   JS=\alpha SJ\mbox{ for some }\alpha\in \R.
\]
Then $\alpha\in\{-1,+1\}$ or $S=0$. Furthermore, if $S=S^T$ and
$\alpha=1$, then $S=c\cdot\mathbf{1}$ for some $c\in\R$.
\end{lemma}
We omit the proof; it is a simple exercise in linear algebra.

\begin{lemma}
\label{LemLinAlg2}
Consider any antisymmetric matrix of the form
\[
   Y:=\left(
      \begin{array}{ccccccc}
         0 & \lambda_1 & & & & & \\
         -\lambda_1 & 0 & & & & & \\
         & & 0 & \lambda_2 & & & \\
         & & -\lambda_2 & 0 & & & \\
         & & & & \ddots & &  \\
         & & & & & 0 & \lambda_k \\
         & & & & & -\lambda_k & 0
      \end{array}
   \right)\in\R^{(2k)\times (2k)},\mbox{ all }\lambda_i\neq 0,\enspace \lambda_i\neq \pm \lambda_j\mbox{ for }i\neq j
\]
(all other entries zero). Let $S=S^T \in \R^{(2k)\times(2k)}$ be any symmetric matrix that commutes with $Y$, i.e.\ $[Y,S]=0$. Then $S$ is a diagonal matrix of the form
\[
   S={\rm diag}(s_1,s_1,s_2,s_2,\ldots,s_k,s_k),\qquad s_i\in\R.
\]
\end{lemma}
\begin{proof}
Define the $2\times 2$ block matrices $\Lambda_i:=\left(\begin{array}{cc} 0 & \lambda_i \\ -\lambda_i & 0\end{array}\right)$, and divide $S$ into $2\times 2$ block matrices $S_{i,j}$:
\[
   S=\left(
      \begin{array}{ccc}
         S_{1,1} & \hdots & S_{1,k} \\
         \vdots & &\vdots \\
         S_{k,1} & \hdots & S_{k,k}
      \end{array}
   \right)
\]
Then the commutator is the symmetric matrix
\[
   [Y,S]=\left(
      \begin{array}{cccc}
         [\Lambda_1,S_{1,1}] & \Lambda_1 S_{1,2} - S_{1,2}\Lambda_2 & \Lambda_1 S_{1,3}-S_{1,3}\Lambda_3 & \hdots \\
         \vdots & [\Lambda_2,S_{2,2}] & \Lambda_2 S_{2,3}-S_{2,3}\Lambda_3 & \vdots \\
         & & \ddots & \\
         & \hdots & & [\Lambda_k,S_{k,k}]
      \end{array}
   \right).
\]
If this is the zero matrix, then $0=[\Lambda_i,S_{i,i}]=-\lambda_i[J,S_{i,i}]$ for all $i$. It follows from Lemma~\ref{LemLinAlg1} that there exists $s_i\in\R$ such that
$S_{i,i}=s_i\cdot\mathbf{1}$. Similarly, for all $i\neq j$, we have $\Lambda_i S_{i,j}=-\lambda_i J S_{i,j} = S_{i,j} \Lambda_j = - S_{i,j} J \lambda_j$, hence
$J S_{i,j} = \alpha S_{i,j} J$ with $\alpha=(\lambda_j/\lambda_i)\not\in\{-1,+1\}$. Thus, Lemma~\ref{LemLinAlg1} yields that $S_{i,j}=0$.
\end{proof}

We show that an analogue of this remains true in odd dimensions:
\begin{lemma}
\label{LemLinAlg3}
Consider any antisymmetric matrix of the form
\[
   Y:=\left(
      \begin{array}{cccccccc}
         0 & \lambda_1 & & & & & & \\
         -\lambda_1 & 0 & & & & & & \\
         & & 0 & \lambda_2 & & & & \\
         & & -\lambda_2 & 0 & & & & \\
         & & & & \ddots & & &  \\
         & & & & & 0 & \lambda_k & \\
         & & & & & -\lambda_k & 0 & \\
         & & & & & & & 0
      \end{array}
   \right)\in\R^{(2k+1)\times (2k+1)},\mbox{ all }\lambda_i\neq 0,\enspace \lambda_i\neq \pm \lambda_j\mbox{ for }i\neq j
\]
(all other entries zero). Let $S=S^T \in \R^{(2k+1)\times(2k+1)}$ be any symmetric matrix that commutes with $Y$, i.e.\ $[Y,S]=0$. Then $S$ is a
diagonal matrix of the form
\[
   S={\rm diag}(s_1,s_1,s_2,s_2,\ldots,s_k,s_k,s_{k+1}),\qquad s_i\in\R.
\]
\end{lemma}
\begin{proof}
Divide $Y$ and $S$ into block matrices:
\begin{eqnarray*}
   Y&=&\left(\begin{array}{cc} \bar Y & 0 \\ 0 & 0 \end{array}\right),\qquad \bar Y\in\R^{(2d)\times(2d)}, \\
   S&=&\left(\begin{array}{cc} S_{1,1} & S_{1,2} \\ S_{1,2}^T & S_{2,2} \end{array}\right),\qquad S_{1,1}\in\R^{(2d)\times(2d)},\enspace S_{1,2}\in\R^{2d},\enspace S_{2,2}\in\R.
\end{eqnarray*}
Then
\[
   [Y,S]=\left(
      \begin{array}{cc}
         [\bar Y,S_{1,1}] & \bar Y S_{1,2} \\
         S_{1,2}^T \bar Y & 0
      \end{array}
   \right).
\]
If this is the zero matrix, then $[\bar Y,S_{1,1}]=0$, and the diagonal form of $S_{1,1}$ with all entries repeated twice follows from Lemma~\ref{LemLinAlg2}.
Furthermore, $S_{1,2}\in{\rm ker}(\bar Y)=\{0\}$. Finally, set $s_{k+1}:=S_{2,2}$.
\end{proof}

Before applying this, we need to show that real quantum mechanics is necessarily equipped with all reversible transformations (conjugations with orthogonal
matrices) to comply with Postulate 2:
\begin{lemma}
For $N\geq 3$, let $\Omega_N$ be the state space of $N$-level real quantum mechanics, and $\G_N$ be a group of reversible transformations on it
such that Postulate 2 is satisfied. Then
\[
   \G_N=\left\{
      \rho\mapsto O\rho O^T\,\,|\,\, O\in \G
   \right\},
\]
where either $\G={\rm SO}(N)$ or $\G={\rm O}(N)$. In particular, $\g_N=\left\{\rho\mapsto [M,\rho]\,\,|\,\, M\in {\rm so}(N)\right\}$.
\end{lemma}
\begin{proof}
Every $G\in \G_N$ is an automorphism of the cone of positive semidefinite symmetric real matrices, and thus of the form $\rho\mapsto Q\rho Q^T$~\cite{Gowda};
preservation of the trace implies that $Q^T Q=\mathbf{1}$, i.e.\ that Q is orthogonal. Define $\G$ as the set of all orthogonal $Q$ such that the map $\rho\mapsto Q \rho Q^T$
is contained in $\G_N$. Clearly $\G$ is a subgroup of ${\rm O}(N)$; since $\G_N$ is topologically closed, so is $\G$.

Now we show that $\G$ contains all of ${\rm SO}(N)$. Let $a,b\in\R$ be irrational numbers such that their difference $a-b$ is also irrational. Define the unit vectors
$e_i:=(0,\ldots,\underbrace{1}_i,0,\ldots,0)^T$, and
\begin{eqnarray*}
   v_1&:=& (\cos(a\pi),-\sin(a\pi),0,\ldots,0)^T,\enspace v_2:=(\sin(a\pi),\cos(a\pi),0,\ldots,0)^T,\enspace v_3=e_3,\ldots,v_N=e_N,\\
   w_1&:=& (\cos(b\pi),-\sin(b\pi),0,\ldots,0)^T,\enspace w_2:=(\sin(b\pi),\cos(b\pi),0,\ldots,0)^T,\enspace w_3=e_3,\ldots,w_N=e_N.
\end{eqnarray*}
Then the sets of vectors $\{v_1,\ldots,v_N\}$ and $\{w_1,\ldots,w_N\}$ are both orthonormal bases of $\R^N$, and so the sets of pure states
$\{|v_1\rangle\langle v_1|,\ldots,|v_N\rangle\langle v_N|\}$ and $\{|w_1\rangle\langle w_1|,\ldots,|w_N\rangle\langle w_N|\}$ are both $N$-frames
in $N$-level real quantum mechanics, and so is $\{ |e_1\rangle\langle e_1|,\ldots,|e_N\rangle\langle e_N|\}$. Thus, according to Postulate 2, there are
two orthogonal matrices $V,W\in\G$ such that
\[
   V|e_i\rangle\langle e_i| V^T = |v_i\rangle\langle v_i|\qquad\mbox{and}\qquad W|e_i\rangle\langle e_i|W^T =|w_i\rangle\langle w_i|\qquad\mbox{for }i=1,\ldots,N.
\]
It follows that there are signs $\sigma_1,\ldots,\sigma_N,\tau_1,\ldots,\tau_N\in\{-1,+1\}$ such that $V|e_i\rangle=\sigma_i|v_i\rangle$ and $W|e_i\rangle=\tau_i|w_i\rangle$.
Hence
\[
   V=\left(
      \begin{array}{ccccc}
         \sigma_1 \cos(a\pi) & \sigma_2 \sin(a\pi) & 0 & \hdots & 0\\
         -\sigma_1 \sin(a\pi) & \sigma_2 \cos(a\pi) & 0 & \hdots & 0\\
         0 & 0 & \sigma_3 & & \\
         \vdots & \vdots & & \ddots & \\
         0 & 0 & & & \sigma_N
      \end{array}
   \right),\qquad
   W=\left(
      \begin{array}{ccccc}
         \tau_1 \cos(b\pi) & \tau_2 \sin(b\pi) & 0 & \hdots & 0\\
         -\tau_1 \sin(b\pi) & \tau_2 \cos(b\pi) & 0 & \hdots & 0\\
         0 & 0 & \tau_3 & & \\
         \vdots & \vdots & & \ddots & \\
         0 & 0 & & & \tau_N
      \end{array}
   \right)
\]
Now we consider two different cases. As the first case, suppose that $\sigma_1=\sigma_2$ or $\tau_1=\tau_2$. Then
\[
   V^2=\left(
      \begin{array}{ccccc}
         \cos(2 a\pi) & \sin(2a\pi) & 0 & \hdots & 0\\
         -\sin(2a\pi) & \cos(2a\pi) & 0 & \hdots & 0\\
         0 & 0 & 1 & & \\
         \vdots & \vdots & & \ddots & \\
         0 & 0 & & & 1
      \end{array}
   \right)\in\G\qquad\mbox{or}\qquad
   W^2 =\left(
      \begin{array}{ccccc}
          \cos(2b\pi) & \sin(2b\pi) & 0 & \hdots & 0\\
         -\sin(2b\pi) & \cos(2b\pi) & 0 & \hdots & 0\\
         0 & 0 & 1 & & \\
         \vdots & \vdots & & \ddots & \\
         0 & 0 & & & 1
      \end{array}
   \right)\in\G.
\]
As the second case, suppose that $\sigma_1\neq \sigma_2$ and $\tau_1\neq  \tau_2$. Then $\sigma:=\sigma_1=-\sigma_2$ and $\tau:=\tau_1=-\tau_2$, and
\[
   VW=\sigma\tau\left(
      \begin{array}{ccccc}
         \cos((a-b)\pi) & \sin((a-b)\pi) & 0 & \hdots & 0\\
         -\sin((a-b)\pi) & \cos((a-b)\pi) & 0 & \hdots & 0\\
         0 & 0 & 1 & & \\
         \vdots & \vdots & & \ddots & \\
         0 & 0 & & & 1
      \end{array}
   \right)\Rightarrow
   (VW)^2=\left(
      \begin{array}{ccccc}
         \cos(2(a-b)\pi) & \sin(2(a-b)\pi) & 0 & \hdots & 0\\
         -\sin(2(a-b)\pi) & \cos(2(a-b)\pi) & 0 & \hdots & 0\\
         0 & 0 & 1 & & \\
         \vdots & \vdots & & \ddots & \\
         0 & 0 & & & 1
      \end{array}
   \right)\in\G.
\]
In both cases, we have established the existence of a matrix in $\G$ that acts as
$\left(\begin{array}{cc}\cos\theta & \sin\theta \\ -\sin\theta & \cos\theta\end{array}\right)$ in the $e_1-e_2$-subspace,
where $\theta$ is an irrational multiple of $\pi$. But any matrix of this form generates all of ${\rm SO}(2)$ by composition and closure.
We can argue similarly for all other $e_i-e_j$-subspaces. The corresponding ${\rm SO}(2)$ rotations in all these planes generate
all special orthogonal matrices, hence ${\rm SO}(N)\subseteq \G$.
\end{proof}

\begin{theorem}
Energy is not observable on any $N$-level real quantum mechanics state space.
\end{theorem}
\begin{proof}
The case $N=1$ is trivial; $N=2$ is shown in the main text, so let $N\geq 3$.
First, consider the case that $N$ is even. Let $H\subset {\rm so}(N)$ be the subspace of matrices
\[
   H:=\left\{
      \left.
      \left(
         \begin{array}{ccccccc}
            0 & \lambda_1 & & & & & \\
            -\lambda_1 & 0 & & & & &\\
            & & 0 & \lambda_2 & & &\\
            & & -\lambda_2 & 0 & & &\\
            & & & & \ddots & & \\
            & & & & &0 & \lambda_{N/2} \\
            & & & & & -\lambda_{N/2} & 0
         \end{array}
      \right)\,\,\right|\,\, \lambda_1,\ldots,\lambda_{N/2}\in\R
   \right\},
\]
and $h\subset \g_N$ be the corresponding subspace of maps of the form $\rho\mapsto [\Lambda,\rho]$ with $\Lambda\in H$. Moreover,
let $H'$ be the set of all $\Lambda\in H$ where the corresponding $\lambda_i$ satisfy $\lambda_i\neq 0$ and $\lambda_i\neq \pm \lambda_j$ for $i\neq j$.
Then $H'$ is dense in $H$. Similarly, by $h'$, denote the set of maps $\rho\mapsto [M,\rho]$ with $M\in H'$; then $h'$ is dense in $h$.

Consider any energy observable assignment $\phi$. Any matrix $M\in {\rm so}(n)$ defines a generator $X\in\g_N$ by $X(\rho):=[M,\rho]$ and vice versa.
This generator is mapped by $\phi$ to some map $\rho\mapsto \tr(\bar M\rho)$, where $\bar M=\bar M^T$. Denote the map $M\mapsto \bar M$ by $\bar\phi$, such that
\[
   [\phi(X)](\rho)=\tr(\bar\phi(M)\rho)\qquad\mbox{if }X(\rho)=[M,\rho].
\]
Then we have the equivalences
\[
   \phi(X)\circ X=0\Leftrightarrow [\phi(X)](X(\rho))=0\mbox{ for all }\rho\Leftrightarrow \tr(\bar\phi(M)[M,\rho])=0\mbox{ for all }\rho\Leftrightarrow
   \tr(\rho[\bar\phi(M),M])=0\mbox{ for all }\rho,
\]
and so $\bar\phi(M)$ is a symmetric matrix that must commute with $M$. Suppose that $M\in H'$, then Lemma~\ref{LemLinAlg3} shows
that $\bar\phi(M)={\rm diag}(s_1,s_1,s_2,s_2,\ldots,s_{N/2},s_{N/2})$. Denote by $\mathcal{H}$ the linear space of all diagonal $(N\times N)$-matrices
of that form. We have shown that $\bar\phi(H')\subset \mathcal{H}$. Since $H'$ is dense in $H$, this implies that $\bar\phi(H)\subset \mathcal{H}$.
Since $\dim H=\dim\mathcal{H}=N/2$, and since $\bar\phi$ is injective, this implies that $\bar\phi(H)=\mathcal{H}$. In particular, there is $0\neq M\in H$
such that $\bar\phi(M)=\mathbf{1}$, so the corresponding generator $X\in h$ satisfies $X(\rho)=[M,\rho]$ which is not identically zero for all $\rho$,
and $[\phi(X)](\rho)=\tr(\rho)=u_N(\rho)$, contradicting the definition of an energy observable assignment.

Now consider the case that $N$ is odd, say, $N=2k+1$. Define the subspace $H$ of antisymmetric matrices by
\[
   H:=\left\{
      \left.\left(
      \begin{array}{cccccccc}
         0 & \lambda_1 & & & & & & \\
         -\lambda_1 & 0 & & & & & & \\
         & & 0 & \lambda_2 & & & & \\
         & & -\lambda_2 & 0 & & & & \\
         & & & & \ddots & & &  \\
         & & & & & 0 & \lambda_k & \\
         & & & & & -\lambda_k & 0 & \\
         & & & & & &  & 0
      \end{array}
   \right)\,\,\right|\,\, \lambda_1,\ldots,\lambda_k\in\R
   \right\}
\]
Similar argumentation as in the even case, now using Lemma~\ref{LemLinAlg3}, shows that $\bar\phi(M)$ is a diagonal matrix for every $M\in H$;
the same conclusion holds true if the subspace $H$ is defined by appending the zero in the top-left corner instead of the bottom-right. But then, by linearity,
the matrix
\[
   M:=\left(
      \begin{array}{cccccccc}
         0 & \lambda_1 & & & & & & \\
         -\lambda_1 & 0 & & & & & & \\
         & & 0 & \lambda_2 & & & & \\
         & & -\lambda_2 & 0 & & & & \\
         & & & & \ddots & & &  \\
         & & & & & 0 & \lambda_k & \\
         & & & & & -\lambda_k & 0& \lambda_k  \\
         & & & & & & -\lambda_k & 0
      \end{array}
   \right)=
   \left(
      \begin{array}{cccccccc}
         0 & \lambda_1 & & & & & & \\
         -\lambda_1 & 0 & & & & & & \\
         & & 0 & \lambda_2 & & & & \\
         & & -\lambda_2 & 0 & & & & \\
         & & & & \ddots & & &  \\
         & & & & & 0 & \lambda_k & \\
         & & & & & -\lambda_k & 0 & \\
         & & & & & &  & 0
      \end{array}
   \right)+
   \left(
      \begin{array}{ccccc}
         0 & & & & \\
           & \ddots & & & \\
           & & 0 & &  \\
           & & & 0 & \lambda_k \\
           & & & -\lambda_k & 0
      \end{array}
   \right)
\]
also has the property that $\bar\phi(M)$ is a diagonal matrix. Suppose that all $\lambda_i\neq 0$, then the only diagonal matrix $S$ that commutes with $M$
is of the form $S={\rm diag}(s_1,s_1,s_2,s_2,\ldots,s_{k-1},s_{k-1},s_k,s_k,s_k)$. Again, arguing analogously to the even case, the subspace of all matrices $M$
of the given form (dropping the condition $\lambda_i\neq 0$) is mapped by $\bar\phi$ injectively into the subspaces of all diagonal matrices $S$ of that form.
Since both are of dimension $k$, there is $M\neq 0$ such that $\bar\phi(M)=\mathbf{1}$, violating the definition of an energy observable assignment.
\end{proof}

\end{appendix}

\end{document}